\documentclass[11pt]{article}
\usepackage{amsmath,amssymb,amsthm,xspace,fullpage,txfonts}
\usepackage[usenames,dvipsnames]{color}
\usepackage[pagebackref,colorlinks=true,urlcolor=blue,linkcolor=blue,citecolor=BrickRed,pdfstartview=FitH]{hyperref}

\newtheorem{theorem}{Theorem}[section]
\newtheorem{lemma}[theorem]{Lemma}
\newtheorem{proposition}[theorem]{Proposition}

\newtheorem{corollary}[theorem]{Corollary}

\theoremstyle{definition}
\newtheorem{definition}[theorem]{Definition}

\newcommand{\nst}[1]{{^*#1}} 
\newcommand{\bbC}{\mathbb{C}}
\newcommand{\bbF}{\mathbb{F}}

\newcommand{\bbN}{\mathbb{N}}
\newcommand{\bbR}{\mathbb{R}}
\newcommand{\bbT}{\mathbb{T}}
\newcommand{\bbU}{\mathbb{U}}
\newcommand{\bbZ}{\mathbb{Z}}

\newcommand{\bfA}{\mathbf{A}}

\newcommand{\bfC}{\mathbf{C}}
\newcommand{\bfF}{\mathbf{F}}
\newcommand{\bfH}{\mathbf{H}}
\newcommand{\bfL}{\mathbf{L}}

\newcommand{\bfX}{\mathbf{X}}
\newcommand{\bfa}{\mathbf{a}}
\newcommand{\bfb}{\mathbf{b}}
\newcommand{\bfd}{\mathbf{d}}
\newcommand{\bff}{\mathbf{f}}
\newcommand{\bfg}{\mathbf{g}}

\newcommand{\bfx}{\mathbf{x}}
\newcommand{\bfy}{\mathbf{y}}
\newcommand{\caA}{\mathcal{A}}
\newcommand{\caB}{\mathcal{B}}

\newcommand{\caF}{\mathcal{F}}

\newcommand{\caL}{\mathcal{L}}
\newcommand{\caP}{\mathcal{P}}

\newcommand{\sfe}{\mathsf{e}}
\newcommand{\st}{\mathrm{st}}

\newcommand{\proj}[2]{#1\mathord{\downharpoonright}_{#2}}

\newcommand{\Aff}{\mathrm{Aff}}

\newcommand{\Fp}{\mathbb{F}_p}
\newcommand{\Ftwo}{\mathbb{F}_2}
\newcommand{\E}{\mathop{\mathbf{E}}}
\newcommand{\bit}{\{0,1\}}
\newcommand{\rank}{\mathrm{rank}}
\newcommand{\depth}{\mathrm{depth}}
\newcommand{\Lovasz}{Lov{\'a}sz\xspace}

\title{Gowers Norm, Function Limits, and Parameter Estimation}
\author{
  Yuichi Yoshida\thanks{
    Supported by JSPS Grant-in-Aid for Young Scientists (B) (No.~26730009), MEXT Grant-in-Aid for Scientific Research on Innovative Areas (24106001), and JST, ERATO, Kawarabayashi Large Graph Project.
  }\\
  National Institute of Informatics and Preferred Infrastructure, Inc.\\
  \texttt{yyoshida@nii.ac.jp}
}

\begin{document}
\maketitle
\begin{abstract}
  Let $\{f_i:\Fp^i \to \bit\}$ be a sequence of functions,
  where $p$ is a fixed prime and $\Fp$ is the finite field of order $p$.
  The limit of the sequence can be syntactically defined using the notion of ultralimit.
  Inspired by the Gowers norm, we introduce a metric over limits of function sequences, and study properties of it.
  One application of this metric is that it provides a characterization of affine-invariant parameters of functions that are constant-query estimable.
  Using this characterization, we show that the property of being a function of a constant number of low-degree polynomials and a constant number of factored polynomials (of arbitrary degrees) is constant-query testable if it is closed under blowing-up.
  Examples of this property include the property of having a constant spectral norm and degree-structural properties with rank conditions.
\end{abstract}

\thispagestyle{empty}
\setcounter{page}{0}
\newpage

\section{Introduction}\label{sec:intro}

Let $p$ be a fixed prime and $\Fp$ be the finite field of order $p$.
For positive integers $n$ and $m$,
an \emph{affine transformation} $A:\Fp^m \to \Fp^n$ is of the form $L+c$,
where $L:\Fp^m \to \Fp^n$ is a linear transformation and $c \in \bbF^n$ is a vector.
When $A$ is injective (in particular, $m \leq n$), we call it \emph{affine embedding}.
The \emph{affine subspace} spanned by an affine transformation $A:\Fp^m \to \Fp^n$ is $\{Ax \mid x \in \Fp^m\}$.
For a function $f:\Fp^n \to \bbR$ and an affine transformation $A:\Fp^m \to \Fp^n$,
we define $f \circ A:\Fp^m \to \bbR$ as $(f\circ A)(x) = f(Ax)$ for all $x \in \Fp^m$.
The \emph{rank} of an affine transformation $A = L + c$, denoted by $\rank(A)$, is defined as the rank of $L$.

Let $\pi$ be a function parameter that maps a function to a value in the range $[0,1]$.
In \emph{parameter estimation} of $\pi$,
given a proximity parameter $\epsilon > 0$, an integer $n \in \bbN$,
and a query access to a function $f:\Fp^n \to \bit$,
we want to approximate $\pi(f)$ to within $\epsilon$ with a probability of at least $2/3$.
We state that parameter $\pi$ is \emph{constant-query estimable} if there is such an algorithm with the number of queries that is independent of $n$ (but may be dependent on $\epsilon$).
We say that a parameter $\pi$ is \emph{affine-invariant} if for any function $f:\Fp^n \to \bit$ and bijective affine transformation $A$, $\pi(f) = \pi(f \circ A)$ holds.
Because we do not want to consider ``unnatural'' parameters such as $\pi(f) = n \pmod 2$, we only consider \emph{oblivious algorithms}~\cite{Bhattacharyya:2010gb,Goldreich:2003zz}, which restrict the input function to a random affine subspace of constant dimension (usually dependent on $\epsilon$) and which then provide an output based solely on that restriction\footnote{From the argument made in~\cite{Goldreich:2003zz}, we can assume that oblivious algorithms does not use internal randomness when making decisions. Further, the non-adaptiveness and uniform choice of affine subspaces are without loss of generality~\cite{Bhattacharyya:2010gb}.}.
Unless stated otherwise, all algorithms considered in this paper are oblivious.
The question of which affine-invariant parameters are obliviously constant-query estimable naturally arises during parameter estimation;
this paper provides a useful characterization of such affine-invariant parameters.
First, however, several notions must be established.

The Gowers norm is a very useful tool for studying the behavior of a function under affine transformation,
For a function $f:\Fp^n \to \bbR$,
the \emph{$d$-th Gowers norm} of $f$ is defined as follows:
\[
  \|f\|_{U^d} := \Bigl|\E_{x, y_1,\ldots,y_d \in \Fp^n}\prod_{I \subseteq [d]} f(x + \sum_{i \in I}y_i) \Bigr|^{1/2^d}.
\]
In that expectation, we take the product of all values of $f$ at every point in a random $d$-dimensional affine subspace.
The Gowers norm is a norm when $d > 1$ and a semi-norm when $d = 1$.
Generally, the $d$-th Gowers norm measures the correlation with polynomials of a degree of at most $d-1$ (more precisely, \emph{non-classical} polynomials~\cite{Tao:2011dw}).
The Gowers norm is used in various areas of theoretical computer science such as constructing pseudorandom generators~\cite{Bogdanov:2010ie},
property testing~\cite{Bhattacharyya:2012ud,Bhattacharyya:2013ii,Hatami:2013ux,Yoshida:2014tq}, coding theory~\cite{Bhowmick:2014un}, and hardness of approximation~\cite{Samorodnitsky:2009di}.

In parameter estimation,
it is important to study the distribution of the input function restricted to a random affine subspace of a constant dimension, say $k$, since an oblivious constant-query algorithm determines the output based on that restriction.
It turns out that two functions $f, g: \Fp^n \to \bit$ have similar distributions if there exists an affine bijection $A:\Fp^n \to \Fp^n$ such that $\|f - g \circ A\|_{U^d}$ is small, where $d$ is an integer dependent on $k$.
With this fact in mind, we can define the distance between $f$ and $g$ as follows:
\[
  \upsilon^d(f, g) := \min_{\substack{A:\Fp^n \to \Fp^n \\ A\text{ is a bijection}}}\|f - g \circ A\|_{U^d}.
\]
Note that $\upsilon^d$ forms a metric space by identifying functions with a distance of zero.

One disadvantage of the distance notion $\upsilon^d(\cdot,\cdot)$ is that the distance between functions on different domains is not defined, and hence, not useful for studying the constant-query estimability of parameters.
This paper's main contribution is the proposal of a distance notion that captures the closeness of the distributions of two functions that are restricted to a random affine subspace of a constant dimension, even if they are defined on different domains.

To define such distance, let us consider the sequence of functions $(f_i:\Fp^i \to \bbR)_{i \in \bbN}$, where $\bbN$ is the set of positive integers.
Since we do not have a distance notion between functions over different domains, we cannot discuss the convergence of the sequence in the usual sense.
However using the notion of \emph{ultralimit} in non-standard analysis, we can syntactically define the limit $\bff:\bfF \to \bbR$ of the sequence $(f_i)$,
where $\bfF$ is the so-called \emph{ultraproduct} of $(\Fp^i)_{i \in \bbN}$.
We call $\bff$ a \emph{function limit} since it is a limit of a function sequence.
We will discuss the properties of $\bfF$ in detail in subsequent sections; what we need to know now is that $\bfF$ is endowed with addition and multiplication as well as a probability measure.
Hence, we can define the $d$-th Gowers norm of $\bff$ as follows:
\[
  \|\bff\|_{U^d} := \Bigl|\int_\bfF \cdots \int_{\bfF} \prod_{I \subseteq [d]}\bff(\bfx + \sum_{i \in I}\bfy_i)  d\bfx d \bfy_1 \cdots d \bfy_\ell\Bigr|^{1/2^d}.
\]
Similarly, we can define the distance between $\bff:\bfF \to \bbR$ and $\bfg:\bfF \to \bbR$ as follows:
\[
  \upsilon^d(\bff, \bfg) := \min_{\substack{\bfA:\bfF \to \bfF \\ \bfA\text{ is a bijection}}}\|\bff - \bfg \circ \bfA\|_{U^d},
\]
where $\bfA$ is over all ultralimits of the sequences of affine bijections.
Again $\upsilon^d(\cdot,\cdot)$ forms a metric space by identifying function limits with a distance of zero.
There is a natural way of identifying a function $f:\Fp^n \to \bbR$ with a function limit; we denote it as $\nst{f}:\bfF \to \bbR$.
With this identification and notion of $\upsilon^d$,
we can discuss the distance between two functions on different domains.
In this paper, we study properties of $\upsilon^d$-metric and give a characterization of constant-query estimable parameters in terms of $\upsilon^d$:
\begin{theorem}\label{the:intro-main}
  An affine-invariant parameter $\pi$ is obliviously constant-query estimable if and only if the following holds:
  For any sequence of functions $(f_i)$ such that $(\nst{f_i})$ converges in the $\upsilon^d$-metric for any $d \in \bbN$, the sequence $\pi(f_i)$ converges.
\end{theorem}

\subsection{Applications to property testing}

Regarding the applicability of Theorem~\ref{the:intro-main},
we consider \emph{property testing}~\cite{Rubinfeld:1996um}, which is a decision version of parameter estimation.
A function $f:\Fp^n \to \bit$ is \emph{$\epsilon$-far} from a property $\caP$ if for any function $g:\Fp^n \to \bit$ satisfying $\caP$,
we have $\Pr_x[f(x) \neq g(x)] \geq \epsilon$.
We say that a property $\caP$ is \emph{constant-query testable} if,
given a proximity parameter $\epsilon > 0$, an integer $n \in \bbN$, and a query access to a function $f:\Fp^n \to \bit$,
with a probability of at least $2/3$,
we can distinguish the case that $f$ satisfies $\caP$ from the case that $f$ is $\epsilon$-far from satisfying $\caP$ with the number of queries independent of $n$ (but may be dependent on $\epsilon$).
A property $\caP$ is \emph{affine-invariant} if,
for any function $f:\Fp^n \to \bit$ satisfying the property $\caP$ and any affine bijection $A:\Fp^n \to \Fp^n$,
$f \circ A$ also satisfies $\caP$.

Note that if the distance to a property $\caP$ (that is, how far from $\caP$) is constant-query estimable,
then $\caP$ is constant-query testable.
For affine-invariant properties,
if a property $\caP$ is constant-query testable,
then the distance to $\caP$ is also constant-query estimable~\cite{Hatami:2013ux}.
Hence Theorem~\ref{the:intro-main} also gives a characterization of constant-query testable affine-invariant properties.
Although another characterization of constant-query testable affine-invariant properties has already been given by the author~\cite{Yoshida:2014tq},
the one given in this paper is much simpler.

Theorem~\ref{the:intro-main} is also useful for showing that a specific property is constant-query testable.
We say that a property $\caP$ is \emph{closed under blowing-up} if, for any $f:\Fp^n \to \bit$ satisfying $\caP$ and an affine transformation $A:\Fp^m \to \Fp^n$ with $m \geq n$ and rank $n$, the function $f \circ A$ satisfies $\caP$.
We call a function $P:\Fp^n \to \Fp$ a \emph{factored polynomial} if it can be written as $P = M \circ A$ for some monomial $M:\Fp^n \to \Fp$ and a bijective affine transformation $A:\Fp^n \to \Fp^n$.
Then, we show the following as an application of Theorem~\ref{the:intro-main}.
\begin{theorem}\label{the:application}
  Let $C,d \in \bbN$ be integers.
  Suppose that a property $\caP$ is closed under blowing-up and every function $f:\Fp^n \to \bit$ satisfying $\caP$ is of the form $f(x) = \Gamma(P_1(x),\ldots,P_c(x),Q_1(x),\ldots,Q_{c'}(x))$ for some $c,c' \leq C$, a function $\Gamma:\Fp^{c+c'} \to \bit$, polynomials $P_i:\Fp^n \to \Fp$ of degree at most $d$, and factored polynomials $Q_i:\Fp^n \to \Fp$ of arbitrary degree.
  Then, the property $\caP$ is obliviously constant-query testable.
\end{theorem}

An important feature of Theorem~\ref{the:application} is that factored polynomials $Q_i$ can have \emph{arbitrary} degrees.
To illustrate, let us consider the \emph{spectral norm} of a function $f:\Ftwo^n \to \bit$, which is the absolute sum of its Fourier coefficients.
The property of having a constant spectral norm is known to be constant-query testable by an ad~hoc argument~\cite{Gopalan:2011jj}.
However, we can easily show its constant-query testability using Theorem~\ref{the:application} since it is closed under blowing-up, and by the Green-Sanders theorem~\cite{Green:2008jf}, if a function has a constant spectral norm, then it can be represented as a function of a constant number of indicators of subspaces, that is, factored polynomials.

Another example captured by Theorem~\ref{the:application} is the following.
\begin{definition}[Ranked degree-structural properties]
  Given an integer $c \in \bbN$, a sequence of integers $\bfd = (d_1,\ldots , d_c) \in \bbZ_+􏰃$, an integer $r \in \bbZ_+$, and a function $\Gamma: \Fp^c \to \bit$, define the \emph{$(c, \bfd, r,\Gamma)$-structured property} to be the collection of functions $f:\Fp^n \to \bit$ of the form $f(x) = \Gamma(P_1(x),\ldots,P_c(x))$, where $P_i:\Fp^n \to \Fp$ is a polynomial of degree at most $d_i$ for each $i \in \{1,\ldots,c\}$ and the rank of $\{P_1,\ldots,P_c\}$ is at least $r$.
\end{definition}
Here, the \emph{rank} of $\{P_1,\ldots,P_c\}$ measures how generic those polynomials are (see Section~\ref{sec:pre} for the definition).
If the rank is high, then those polynomials do not show unexpected behavior in many situations, and hence the rank is regarded as an important notion in higher order Fourier analysis.
It is obvious that a ranked degree-structural property is closed under blowing-up, and hence constant-query testability follows from Theorem~\ref{the:application}.

Bhattacharyya~\emph{et~al.}~\cite{Bhattacharyya:2013ii} showed that the special case of $r = 0$ is constant-query testable with one-sided error.
We note that, when $r = 0$, a ranked degree-structural property $\caP$ is \emph{affine-subspace hereditary}, that is,
for any function $f:\Fp^n \to \bit$ satisfying $\caP$ and affine embedding $A:\Fp^k \to \Fp^n\;(k \leq n)$, $f \circ A$ also satisfies $\caP$,
and their argument crucially uses this property.
However, when $r > 0$, a ranked degree-structural property is not affine-subspace hereditary anymore, and we cannot extend their argument.
Actually we cannot hope for one-sided error testers for properties that are not affine-subspace hereditary~\cite{Bhattacharyya:2010gb}.

\subsection{Related work}

\paragraph{Testing affine-invariant properties of functions:}
Rubinfeld and Sudan~\cite{Rubinfeld:1996um} introduced the notion of property testing; since then, a lot of function properties have been shown to be constant-query testable.
Refer to~\cite{Ron:2010ua,Rubinfeld:2011ik}; a full length book is also available~\cite{Goldreich:2011cg}.
In a celebrated work, Blum~\emph{et~al.}~\cite{Blum:1993cn} showed that linearity is constant-query testable.
Then, Alon~\emph{et~al.}~\cite{Alon:2005jl} extended that result by showing that low-degree polynomials are constant-query testable, and tight query complexity was achieved by Bhattacharyya~\emph{et~al.}~\cite{Bhattacharyya:2010kw}.
Along with the recent development of higher order Fourier analysis~\cite{Green:2008iz,Tao:2011dw,Kaufman:2008ff},
there has been rapid progress in characterizing constant-query testable affine-invariant properties.
Bhattacharyya~\emph{et~al.}~\cite{Bhattacharyya:2012ud,Bhattacharyya:2013ii} showed that every locally characterized property is constant-query testable,
which almost characterizes affine-invariant properties that are constant-query testable with one-sided error.
As we have mentioned, Hatami and Lovett~\cite{Hatami:2013ux} showed that the distance to any constant-query testable affine-invariant property is constant-query estimable.
Finally, the author~\cite{Yoshida:2014tq} obtained a characterization of constant-query testable affine-invariant properties.
Although non-standard analysis is used to show the Gowers inverse theorem~\cite{Tao:2011dw},
every previous work on property testing used the theorem as a black box.
In particular, the characterization given in~\cite{Yoshida:2014tq} does not involve the notion of ultralimits (though the characterization itself is complicated).

\paragraph{Graph limits:}
\Lovasz and Szegedy~\cite{Lovasz:2006jj} introduced the notion of a graph limit, called a \emph{graphon}.
Let $G$ be a graph on $n$ vertices.
Then, $G$ can be seen as a $\bit$-valued function over $[0,1]^2$.
For any $i, j \in [0,1]$ that are not multiples of $1/n$, $G(i,j)$ is equal to one if and only if the vertices $\lceil ni\rceil$ and $\lceil nj\rceil$ are adjacent (we can define the rest of $G$ arbitrarily since they have measures of zero).
In~\cite{Lovasz:2006jj} and subsequent works~\cite{Borgs:2006el,Lovasz:2010dv,Borgs:2006gc},
the properties of graphons and an associated norm, called the \emph{cut-norm}, are studied.
See~\cite{Lovasz:2012wn} for a book on this subject.
In particular, a characterization of constant-query estimable parameters of a graph is shown in~\cite{Borgs:2006el}.
We note that a graphon is a conceptually simpler notion than a function limit since we do not have to resort to ultralimits and since the cut norm does not involve a parameter, unlike the Gowers norm.

\paragraph{Function limits:}
Recently, Hatami~\emph{et~al.}~\cite{Hatami2014:tj} introduced another notion of function limits.
They showed that any sequence of functions such that the distributions obtained by restricting them to a random affine subspace of constant dimension converge can be represented as a function limit and vice versa.
Using their definition, however, it is unclear how to define the distance between function limits and hence functions over different domains.
In particular, we were unable to exploit their notion to study parameter estimation.


\subsection{Organization}
We introduce notions and definitions from higher order Fourier analysis as well as the theory of ultralimits in Section~\ref{sec:pre}.
In Section~\ref{sec:general-tl}, we formally define the Gowers norm for function limits and related notions, whose properties are also studied in that section.
In Section~\ref{sec:metric}, we introduce the $\upsilon^d$-metric and show several of its properties.
We give a characterization of constant-query estimable affine-invariant parameters in Section~\ref{sec:characterization}, and show applications in Section~\ref{sec:applications}.

\section{Preliminaries}\label{sec:pre}
For an integer $n$, $[n]$ denotes the set $\{1,2,\ldots,n\}$.
Let $\bbR_+$ be the set of non-negative real numbers and $\overline{\bbR} = \bbR \cup \{-\infty,\infty\}$.
We denote the set of all affine bijections from $\Fp^n$ to itself as $\Aff(\Fp)$.
For real values $a, b$, and $c$, $a = b \pm c$ means that $b - c \leq a \leq b + c$.

\subsection{Higher order Fourier analysis over $\Fp$}
We review notions from higher order Fourier analysis.
Most of the material in this section is directly quoted from~\cite{Bhattacharyya:2013ii,Hatami:2013ux,Yoshida:2014tq}.
See~\cite{Tao:2012ue} for further details.

\subsubsection{Uniformity norms and non-classical polynomials}

\begin{definition}[Multiplicative derivative]
  Given a function $f:\Fp^n \to \bbC$, and an element $h \in \Fp^n$,
  the \emph{multiplicative derivative} in direction $h$ of $f$ is the function $\Delta_hf:\Fp^n \to \bbC$ satisfying $\Delta_h f(x) = f(x+h)\overline{f(x)}$ for all $x \in \Fp^n$.
\end{definition}

\begin{definition}[Gowers norm]
  Given a function $f:\Fp^n \to \bbC$ and an integer $d \in \bbN$, the $d$-th Gowers norm of $f$ is as follows:
  \begin{align*}
    \|f\|_{U^d} := \left|\E_{x,y_1,\ldots,y_d \in \Fp^n}[(\Delta_{y_1}\Delta_{y_2} \cdots \Delta_{y_d}f)(x)] \right|^{1/2^d}.
  \end{align*}
\end{definition}
Note that,
as $\|f\|_{U^1} = |\E[f]|$, the first Gowers norm is only a semi-norm.
However for $d > 1$, $\|\cdot\|_{U^d}$ is indeed a norm.

The following lemma connects the Gowers and $L_1$ norms.
\begin{lemma}[Claim 2.21 of~\cite{Hatami:2013ux}]\label{lem:bound-gowers-norm-by-l1-norm}
  Let $f:\Fp^n \to [-1,1]$. For any $d \in \bbN$, we have
  \[
    \|f\|_{U^d} \leq \|f\|_1^{1/2^d}.
  \]
\end{lemma}

If $f = e^{2\pi i P/p}$ for a polynomial $P: \Fp^n \to \Fp$  of a degree less than $d$,
then $\|f\|_{U^d} = 1$ holds.
If $d < p$ and $\|f\|_{\infty} \leq 1$,
then in fact,
the converse holds,
meaning that any function $f: \Fp^n \to \bbC$ satisfying $\|f\|_\infty \leq 1$ and $\|f\|_{U^d} = 1$ is of this form.
But when $d \geq p$,
the converse is no longer true.
To characterize functions $f:\Fp^n \to \bbC$ with $\|f\|_\infty \leq 1$ and $\|f\|_{U^d} = 1$, we define the notion of non-classical polynomials.

Non-classical polynomials might not be necessarily $\Fp$-valued.
Some notation needs to be introduced.
Let $\bbT$ denote the circle group $\bbR / \bbZ$.
This is an abelian group with group operation denoted by $+$.
For an integer $k \geq 0$,
let $\bbU_k$ denote $\frac{1}{p^k}\bbZ / \bbZ$, a subgroup of $\bbT$.
Let $\iota : \Fp \to \bbU_1$ be the injection $x \mapsto \frac{|x|}{p} \bmod 1$,
where $|x|$ is the standard map from $\Fp$ to $\{0,1,\ldots,p-1\}$.
Let $\sfe:\bbT \to \bbC$ denote the character $\sfe(x) = e^{2\pi i x}$.

\begin{definition}[Additive derivative]
  Given a function $P : \Fp^n \to \bbT$ and an element $h \in \Fp^n$,
  the \emph{additive derivative} in direction $h$ of $f$ is the function $D_hP : \Fp^n \to \bbT$ satisfying $D_hP(x)=P(x+h)-P(x)$ for all $x \in \Fp^n$.
\end{definition}
\begin{definition}[Non-classical polynomials]
  For an integer $d \in \bbN$,
  a function $P : \Fp^n \to \bbT$ is said to be a \emph{non-classical polynomial} of a degree of at most $d$ (or simply a polynomial of a degree of at most $d$) if for all $x,y_1,\ldots,y_{d+1} \in \Fp^n$, it holds that
  \begin{align*}
    (D_{y_1} \cdots D_{y_{d+1}}P)(x) = 0.
  \end{align*}
  The degree of $P$ is the smallest $d$ for which the above holds.
  A function $P : \Fp^n \to \bbT$ is said to be a \emph{classical polynomial} of a degree of at most $d$ if it is a non-classical polynomial of a degree of at most $d$ whose image is contained in $\iota(\Fp)$.
\end{definition}
It is a direct consequence that a function $f : \Fp^n \to \bbC$ with $\|f\|_\infty \leq 1$ satisfies $\|f\|_{U^{d+1}} = 1$ if and only if $f = \sfe(P)$ for a non-classical polynomial $P : \Fp^n \to \bbT$ of a degree of at most $d$.

\begin{lemma}[Lemma 1.7 in~\cite{Tao:2011dw}]
  A function $P : \Fp^n \to \bbT$ is a polynomial of a degree of at most $d$ if and only if $P$ can be represented as follows:
  \begin{align*}
    P(x_1,\ldots,x_n) = \alpha + \sum_{\substack{0\leq d_1,\ldots,d_n < p; h \geq 0: \\ 0 < \sum_i d_i \leq d - h(p-1)}} \frac{c_{d_1,\ldots,d_n,h}|x_1|^{d_1} \cdots |x_n|^{d_n}}{p^{h+1}} \bmod 1,
  \end{align*}
  for a unique choice of $c_{d_1,\ldots,d_n,h} \in \{0,1,\ldots,p-1\}$ and $\alpha \in \bbT$.
  The element $\alpha$ is called the \emph{shift} of $P$,
  and the largest integer $h$ such that there exist $d_1,\ldots,d_n$ for which $c_{d_1,\ldots,d_n,h} \neq 0$ is called the \emph{depth} of $P$.
  Classical polynomials correspond to polynomials with 0 shift and 0 depth.
\end{lemma}

The degree and depth of a polynomial $P$ is denoted by $\deg(P)$ and $\depth(P)$, respectively.
Also, for the convenience of exposition, we will assume throughout this paper that the shifts of all polynomials are zero, which does not affect any of the results in this work.
Hence, any polynomial of depth $h$ takes values in $\bbU_{h+1}$.

\subsubsection{Rank and decomposition theorem}
We will often need to study the Gowers norms of exponentials of polynomials.
As described below, if this analytic quantity is non-negligible, then there is an algebraic explanation: it is possible to decompose the polynomial as a function of a constant number of low-degree polynomials.
To state this rigorously, let us define the notion of the rank of a polynomial.

\begin{definition}[Rank of a polynomial]
  Given a polynomial $P : \Fp^n \to \bbT$ and an integer $d > 1$, the \emph{$d$-rank} of $P$, denoted as $\rank_d(P)$, is defined to be the smallest integer $r$ such that there exist polynomials $Q_1,\ldots,Q_r : \Fp^n \to \bbT$ of degrees $\leq d - 1$ and a function $\Gamma : \bbT^r \to \bbT$ satisfying $P(x) = \Gamma(Q_1(x), \ldots , Q_r(x))$.
  If $d = 1$, then $1$-rank is defined to be $\infty$ if $P$ is non-constant and $0$ otherwise.
  The \emph{rank} of a polynomial $P : \Fp^n \to \bbT$ is its $\deg(P)$-rank.
\end{definition}

Note that for an integer $\lambda \in [1, p - 1]$, $\rank(P ) = \rank(\lambda P )$.
The following theorem shows that a high-rank polynomial has a small Gowers norms.

\begin{theorem}[Theorem~1.20 of~\cite{Tao:2011dw}]\label{the:small-gowers-norm->high-rank}
  For any $\epsilon > 0$ and integer $d \in \bbN$,
  there exists an integer $r = r_{\ref{the:small-gowers-norm->high-rank}}(\epsilon,d)$ such that the following holds.
  For any polynomial $P : \Fp^n \to \bbT$ of degree at most $d$,
  if $\rank_d(P) \geq r$,
  then $\|\sfe(P)\|_{U^d} \leq \epsilon$.
\end{theorem}


Now we introduce the notion of a factor.
Note that a polynomial sequence $(P_1,\ldots,P_C)$ on $m$ variables of depth $(h_1,\ldots,h_C)$ defines a partition of the space $\prod_{i=1}^C \bbU_{h_i+1}$.
That is, for any tuple $(b_1,\ldots,b_C)$ with $b_i \in \bbU_{h_i+1}$ for each $i \in \{1,\ldots,C\}$,
there is a corresponding part, called an \emph{atom}, $\{x \in \Fp^m \mid (P_1(x),\ldots,P_C(x)) = (b_1,\ldots,b_C)\}$.
We call the partition the \emph{factor}, defined by $(P_1,\ldots,P_C)$ and denoted by $\caB(P_1,\ldots,P_C)$.

The \emph{complexity} of $\caB$, denoted $|\caB|$, is the number of defining polynomials $C$.
The degree of $\caB$ is the maximum degree among its defining polynomials $P_1, \ldots , P_C$.
If $P_1, \ldots , P_C$ are of depths $h_1, \ldots, h_C$, respectively,
then $\|\caB\| = \prod_{i=1}^C p^{h_i+1}$ is called the \emph{order} of $\caB$.
Notice that the number of atoms of $\caB$ is bounded by $\|\caB\|$.

Next, we  formalize the notion of the rank for a generic collection of polynomials.
Intuitively, this should mean that there are no unexpected algebraic dependencies among the polynomials.
\begin{definition}[Rank and regularity]
  A polynomial factor $\caB$ defined by a sequence of polynomials $P_1,\ldots,P_C : \Fp^n \to \bbT$ with respective depths $h_1,\ldots,h_C$ is said to have \emph{rank $r$} if $r$ is the smallest integer for which there exist $(\lambda_1,\ldots,\lambda_C) \in \bbZ^C$ so that $(\lambda_1 \bmod p^{h_1+1},\ldots, \lambda_C \bmod p^{h_C+1}) \neq (0,\ldots,0)$ and the polynomial $Q = \sum_{i=1}^C \lambda_i P_i$ satisfies $\rank_d(Q) \leq r$ where $d = \max_i \deg(\lambda_i P_i)$.

  The \emph{rank} of a polynomial sequence $P_1,\ldots,P_C$, denoted as $\rank(P_1,\ldots,P_C)$, is the rank of the factor $\caB(P_1,\ldots,P_C)$.
  Given a polynomial factor $\caB$ and a function $r : \bbN \to \bbN$,
  we say that $\caB$ is \emph{$r$-regular} if $\caB$ is of a rank of at least $r(|\caB|)$.
\end{definition}
If the rank of a polynomial factor is high, then each atom has almost the same size~\cite{Bhattacharyya:2013ii}.
However, we do not state it here formally, since it will not be used in this paper.



The following decomposition theorem is one of the main tools in higher order Fourier analysis.

\begin{theorem}[Decomposition theorem]\label{the:decomposition}
  Suppose $\delta > 0$ and $d \in \bbN$ is an integer.
  Let $\eta : \bbN \to \bbR_+$ be an arbitrary non-increasing function and $r : \bbN \to \bbN$ be an arbitrary non-decreasing function.
  Then there exists $C = C_{\ref{the:decomposition}}(\delta, \eta,d,r)$ such that the following holds.

  Given $f : \Fp^n \to \bit$, there exist three functions $f_1, f_2, f_3 : \Fp^n \to \bbR$ and a polynomial factor $\caB$ of a degree of at most $d$ and a complexity of at most $C$ such that the following conditions hold:
  \begin{enumerate}
  \setlength{\itemsep}{0pt}
  \item $f = f_1 + f_2 + f_3$.
  \item $f_1 = \E[f \mid \caB]$, that is, $f_1$ is obtained from $f$ by averaging each atom.
  \item $\|f_2\|_2 \leq \delta$.
  \item $\|f_3\|_{U^{d+1}} \leq  \eta(|\caB|)$.
  \item $f_1$ and $f_1 + f_3$ have range $[0, 1]$;
    $f_2$ and $f_3$ have range $[-1, 1]$.
  \item $\caB$ is $r$-regular.
  \end{enumerate}
\end{theorem}

\subsubsection{Systems of linear forms}
A \emph{linear form} in $k$ variables is a vector $L = (\lambda_1,\ldots,\lambda_k) \in \Fp^k$, which is regarded as a linear function from $V^k$ to $V$ for any vector space $V$ over $\Fp$:
If $x = (x_1,\ldots,x_k) \in V_k$, then $L(x) := \lambda_1 x_1 + \cdots + \lambda_k x_k$.
A linear form $L = (\lambda_1,\lambda_2,\ldots,\lambda_k)$ is said to be \emph{affine} if $\lambda_1 = 1$.
A \emph{system of linear forms} in $k$ variables is a finite set $\caL \subseteq \Fp^k$ of linear forms in $k$ variables.
A system of linear forms is called \emph{affine} if it comprises affine linear forms.

Given a function $f : \Fp^n \to \bbC$ and a system of linear forms $\caL = \{L_1,\ldots,L_m\} \subseteq \Fp^k$, define
\[
  t_\caL(f) := \E_{x_1,\ldots,x_k}\Bigl[\prod_{L \in \caL}f(L(x_1,\ldots,x_k)) \Bigr].
\]
Note that for any function $f:\Fp^n \to \bbR$, affine bijection $A:\Fp^n \to \Fp^n$, and affine system of linear forms $\caL$, we have $t_\caL(f) = t_\caL(f \circ A)$.
Also, by choosing $\caL_d := \{L_I \in \Fp^{d+1}: I \subseteq [d]\}$ for $L_I(x_0,x_1,\ldots,x_d) := x_0 + \sum_{i \in I}x_i$,
we have $\|f\|_{U^d} = |t_{\caL_d}(f)|^{1/2^d}$.

\begin{definition}\label{def:true-complexity}
  A system of linear forms $\caL = \{L_1, \ldots , L_m\} \subseteq \Fp^k$ is said to be of \emph{true complexity} at most $d$ if there exists a function $\delta : \bbR_+ \to \bbR_+$ such that $\lim\limits_{\epsilon\to0} \delta(\epsilon) = 0$ and
  \[
    \left|\E_{x_1,\ldots,x_k}\Bigl[\prod_{i=1}^m f_i(L_i(x_1,\ldots,x_k))\Bigr] \right| \leq \min_i \delta(\|f_i\|_{U^{d+1}})
  \]
  holds for all $f_1,\ldots,f_m:\Fp^n \to [-1,1]$ .
\end{definition}
The true complexity of an affine system of $m$ linear forms is at most $mp$~\cite{Gowers:2009cv}.


The following lemma states that, if $f$ and $g$ are close in the sense that $f-g$ has a small $d$-th Gowers norm, then we cannot distinguish them in terms of $t_\caL$, where $\caL$ is a system of linear form with true complexity $d$.
\begin{lemma}\label{lem:tl-difference-by-gowers-norm-difference}
  Let $\caL = \{L_1,\ldots,L_m\} \subseteq \Fp^k$ be a system of linear forms of true complexity of at most $d$.
  Then, there exists a function $\delta:\bbR_+ \to \bbR_+$ such that $\lim\limits_{\epsilon \to 0}\delta(\epsilon) = 0$ and
  \[
    |t_\caL(f) - t_\caL(g)| \leq \delta(\|f - g\|_{U^{d+1}})
  \]
  holds for any $f, g:\Fp^n \to [0, 1]$.
\end{lemma}
\begin{proof}
  We write $t_\caL(f) - t_\caL(g)$ as a telescopic sum
  \[
    t_\caL(f) - t_\caL(g) = \sum_{i \in [m]} \E_{x \in (\Fp^n)^k } \prod_{j < i}f(L_j(x)) \cdot (f(L_i(x)) - g(L_i(x))) \cdot \prod_{j>i}g(L_j(x)).
  \]
  We bound each term in the sum.
  From the definition of true complexity,
  \[
    \Bigl| \E_{x \in (\Fp^n)^k } \prod_{j < i}f(L_j(x)) \cdot (f(L_i(x)) - g(L_i(x))) \cdot \prod_{j>i}g(L_j(x)) \Bigr| \leq \delta'(\|f - g\|_{U^{d+1}}),
  \]
  where $\delta'$ is from Definition~\ref{def:true-complexity}.
  Then, we have $|t_\caL(f) - t_\caL(g)| \leq m \delta'(\|f - g\|_{U^{d+1}})$.
  By setting $\delta(\epsilon) := m \delta'(\epsilon)$,
  we have the lemma.
\end{proof}


\subsection{Non-standard analysis}
We now review the theory of ultralimits, or \emph{non-standard analysis}.
Most of the material in this section is found in~\cite{Warner:2012tc}.

An \emph{ultrafilter} on $\bbN$ is a set $\omega$ comprising subsets of $\bbN$ satisfying the following conditions:
\begin{itemize}
\itemsep=0pt
\item The empty set does not lie in $\omega$.
\item If $A \subseteq \bbN$ lies in $\omega$, then any subset of $\bbN$ containing $A$ lies in $\omega$.
\item If $A$ and $B$ lie in $\omega$, then the intersection $A \cap B$ lies in $\omega$.
\item If $A \subseteq \bbN$, then exactly one of $A$ and $\bbN \setminus A$ lies in $\omega$.
\end{itemize}
Furthermore, if no finite set lies in $\omega$, then we say that $\omega$ is a \emph{non-principal ultrafilter}.
A non-principal filter exists and, in what follows, we fix a non-principal filter $\omega$.

An ultraproduct $\bfA$ of a sequence of sets $(A_i)_{i \in \bbN}$ with respect to $\omega$ is defined as follows.
First construct the Cartesian product $\prod_{i\in \bbN}A_i$.
Define an equivalence relation $\bfa \sim \bfb$, where $\bfa = (a_1,a_2,\ldots)$ and $\bfb = (b_1,b_2,\ldots)$, by
\[
  \bfa \sim \bfb \Leftrightarrow \{i \in \bbN: a_i = b_i \} \in \omega.
\]
Then let $\bfA = \prod_{i \in \bbN}A_i / \sim$.
One can think of $\bfA$ as a sort of completion where one can take the limit of arbitrary sequences, rather than just Cauchy sequences:
given a  sequence $\{a_i\}_{i \in \bbN}$, the equivalence class in $\bfA$ of this sequence will be denoted as follows:
\[
  \bfa = \lim_{i \to \omega}a_i.
\]

Thus, in this terminology, we have $\lim\limits_{i\to \omega} a_i = \lim\limits_{i \to \omega} b_i$ if and only if the set of $i \in \bbN$ such that $a_i = b_i$ is a member of $\omega$.
Similarly, for subsets $H_i \subseteq A_i$ we denote by $\bfH$ or $\lim\limits_{i \to \omega} H_i$ the set of all elements of the ultraproduct arising from limits of points in the given subsets:
\[
  \lim_{i \to \omega} H_i = \Bigl\{\lim_{i \to \omega}a_i : a_i \in H_i, i \in \bbN\Bigr\}.
\]
Such sets are called \emph{internal sets}.

If all of $A_i$ are the same space, the ultraproduct is called an \emph{ultrapower}.
Ultrapowers with respect to a non-principal ultrafilter will be denoted with a prior asterisk;
for example, the ultrapowers of $\bbN$ and $\bbR$ are written $\nst{\bbN}$ and $\nst{\bbR}$, respectively.
The latter object is called the set of \emph{hyperreal numbers}.
The order structure carries over into the hyperreals: for real sequences $(a_i)$ and $(b_i)$ whose ultralimits are $\bfa$ and $\bfb$, respectively, exactly one of the sets $\{i:a_i <b_i\}$, $\{i:a_i =b_i\}$, or $\{i:a_i >b_i\}$ is in $\omega$.
In the first case we say $\bfa < \bfb$, in the second $\bfa = \bfb$, and in the third $\bfa > \bfb$.

We will assume basic facts about $\nst{\bbN}$ and hyperreals, which can be found in~\cite{Robinson:1996uq}: call a hyperreal \emph{standard} if it can be written as $\lim\limits_{i \to \omega}r$ for some constant $r \in \bbR$; thus the reals can be considered a subset of the hyperreals (and likewise for $\nst{\bbN}$).
The hyperreals are an ordered field with an ordering extending that of the reals.
Define an absolute value in the obvious way, by setting
\[
  \Bigl|\lim_{i\to \omega} r_i\Bigr| = \lim_{i\to \omega} |r_i|,
\]
which will be a nonnegative hyperreal.
We call $\bfa \in \nst{\bbR}$ \emph{bounded} if $|\bfa| < \bfC$ for some standard $\bfC$, and we call $\bfa$ \emph{infinitesimal} if $|\bfa| < \bfC$ for all standard $\bfC$.
Hyperreals that are not bounded are called \emph{infinite}.
Every bounded $\bfa$ has a unique decomposition
\[
  \bfa = \st(\bfa) + (\bfa - \st(\bfa))
\]
into a standard part $\st(\bfa)$ and an infinitesimal part $\bfa - \st(\bfa)$, where the mapping $\bfa \mapsto \st(\bfa)$ is a homomorphism from the ring of bounded hyperreals to the reals.

Given a sequence of functions $(f_i:A_i \to \bbR)$, we can form an \emph{ultralimit} $\bff = \lim\limits_{i \to \omega}f_i:\bfA \to \nst{\bbR}$ by defining
\[
  \bff\Bigl(\lim_{i \to \omega}x_i\Bigr) = \lim_{i \to \omega}f_i(x_i).
\]
In what follows, we study the standard part of the ultralimit of a sequence of functions.
For a sequence of functions $(f_i)$, the function $\bff:\bfA \to \overline{\bbR}$ that is defined as $\bff = \st(\lim\limits_{i \to \omega}f_i)$ is called the \emph{function limit}\footnote{This term is not standard in non-standard analysis.} of $(f_i)$.

Suppose that each $A_i$ is an abelian group equipped with a normalized measure $\mu_i$ such that the measure spaces formed are compatible with the group structure in the sense that the action of the group on any measurable set is again measurable.
Then, there is a normalized measure on $\bfA$ called the \emph{Loeb measure} (see~\cite{Warner:2012tc} for its construction.)
\begin{lemma}[Lemma~3.6 of~\cite{Warner:2012tc}]
  Let $(f_i:A_i \to \bbR)$ be a sequence of $\mu_i$-measurable functions on $A_i$ for each $i \in \bbN$ and let $\bff = \st(\lim\limits_{i \to \omega}f_i)$ be its function limit.
  Then, $\bff$ is $\mu$-measurable, where $\mu$ is the Loeb measure on $\bfA$.
\end{lemma}
A partial converse holds:
\begin{lemma}[Proposition~3.8 of~\cite{Warner:2012tc}]\label{lem:measurable-ultra-function->sequence-of-measurable-function}
  For every $\mu$-measurable function $\bfg: \bfA \to \overline{\bbR}$, there exists a sequence of $\mu_i$-measurable functions $(f_i: A_i \to \bbR)$ such that, for $\bff = \st(\lim\limits_{i\to \omega}f_i)$, we have $\bff = \bfg$ almost everywhere with respect to $\mu$.
  Furthermore, if $\bfg$ is bounded, then the $f_i$ can be chosen so as to be uniformly bounded (above or below) with the same bound.
\end{lemma}
Given a $\mu$-measurable $\bfg: \bfA \to \bbR$, we will call the sequence $(f_i:A_i \to \bbR)$ given by Lemma~\ref{lem:measurable-ultra-function->sequence-of-measurable-function} a \emph{lifting} of $\bfg$.
A lifting will be highly non-unique in general.
However, the following two relations hold between $\bfg$ and $\bff = \st(\lim\limits_{i \to \omega}f_i)$.
\begin{lemma}[Proposition~3.9 of~\cite{Warner:2012tc}]\label{lem:exchange-limit-and-int}
  Let $(f_i:A_i \to \bbR)$ be a sequence of uniformly bounded $\mu_i$-measurable functions and $\bff = \st(\lim\limits_{i\to \omega}f_i)$.
  Then,
  \[
    \int_\bfA \bff(x) dx = \st \Bigl( \lim_{i\to \omega}\int_{A_i}f_i(x) d x \Bigr).
  \]
\end{lemma}

\begin{lemma}[Proposition~3.10 of~\cite{Warner:2012tc}]\label{lem:exchange-limit-and-int-sub}
  Suppose $\bff:\bfA \to \overline{\bbR}$ be $\mu$-measurable and bounded, and let $(f_i:A_i \to \bbR)$ be any bounded lifting of $\bff$.
  Then,
  \[
    \int_\bfA \bff(x) dx = \st\Bigl(\lim_{i\to \omega}\int_{A_i}f_i(x) dx \Bigr).
  \]
\end{lemma}

In this paper, we only consider the case that $A_i = \Fp^i$ for each $i \in \bbN$.
So we define $\bfF = \lim\limits_{i \to \omega}\Fp^i$.
Let $\mu_i$ be the normalized counting measure on $\Fp^i$ for each $i \in \bbN$, and let $\mu$ be the corresponding Loeb measure.
Let $\caF$ be the set of uniformly bounded $\mu$-measurable functions of the form $\bff: \bfF \to \bbR$.
Let $\caF_{\bit}$ and $\caF_{[0,1]}$ be the sets of $\mu$-measurable functions of the form $\bff: \bfF \to \bit$ and $\bff:\bfF \to [0,1]$, respectively.

\section{Generalization of $t_\caL$}\label{sec:general-tl}

Let $\caL$ be a system of linear forms in $k$ variables.
Although the notion of $t_\caL$ was originally defined over functions, we can generalize it to function limits using the Loeb measure $\mu$.
That is, for $\bff:\bfF \to \bbR$, we define
\[
  t_\caL(\bff) := \int_\bfF \cdots \int_{\bfF} \prod_{L \in \caL}\bff(L(\bfx_1,\ldots,\bfx_k))  d \bfx_1 \cdots d \bfx_k.
\]
Since a Fubini-type theorem, called Keisler’s Fubini theorem, holds for the measure $\mu$, the value $t_\caL(\bff)$ is uniquely determined regardless of the order of taking integrations.

We define the \emph{$d$-th Gowers norm} of $\bff$ as follows:
\[
  \|\bff\|_{U^d} := |t_{\caL_{d}}(\bff)|^{1/2^d} = \Bigl|\int_\bfF \cdots \int_{\bfF} \prod_{I \subseteq [d]}\bff(\bfx + \sum_{i \in I}\bfy_i)  d\bfy d \bfy_1 \cdots d \bfy_d\Bigr|^{1/2^d}.
\]
Again $\|\cdot\|_{U^1}$ is only a semi-norm, but $\|\cdot\|_{U^d}$ for $d \geq 2$ is indeed a norm.

The following lemma states that we can exchange $\st(\lim(\cdot))$ inside and outside of $t_\caL(\cdot)$.
\begin{lemma}\label{lem:tL-exchange}
  Let $\bff \in \caF$, $f_i$ be its lifting, and $\caL$ be a system of linear forms.
  Then,
  \[
    t_\caL(\bff) = \st(\lim_{i \to \omega} t_\caL(f_i)).
  \]
\end{lemma}
\begin{proof}
  It holds that
  \begin{align*}
    t_\caL(\bff)
    & = \int_\bfF \cdots \int_{\bfF} \prod_{L \in \caL}\bff (L(\bfx^1,\ldots,\bfx^k)) d \bfx^1\cdots d \bfx^k \\
    & = \int_\bfF \cdots\int_{\bfF} \prod_{L \in \caL} \st(\lim_{i\to \omega}f_i (L(x^1_i,\ldots,x^k_i))) d\bfx^1\cdots d\bfx^k \tag{$\bfx^j =: \lim\limits_{i \to \omega}x^j_i$} \\
    & = \int_\bfF \cdots\int_{\bfF} \st(\lim_{i\to \omega}\prod_{L \in \caL} f_i (L(x^1_i,\ldots,x^k_i))) d\bfx^1\cdots d\bfx^k \\
    & = \st\Bigl(\lim_{i\to \omega}\int_\bfF \cdots \int_{\bfF} \prod_{L \in \caL} f_i (L(x^1_i,\ldots,x^k_i)) d\bfx^1\cdots d\bfx^k\Bigr) \tag{by Lemma~\ref{lem:exchange-limit-and-int}}\\
    & = \st (\lim_{i \to \omega} t_L(f_i))
    \qedhere
  \end{align*}
\end{proof}

Let $\bfA = \lim\limits_{i \to \omega}A_i$, where $A_i:\Fp^i \to \Fp^i$ is an affine transformation for each $i \in \bbN$.
For $\bfx = \lim\limits_{i \to \omega}x_i \in \bfF$, we define $\bfA \bfx = \lim\limits_{i \to \omega}A_ix_i$.
Hence, $\bfA$ can be seen as a map from $\bfF$ to itself, and we call $\bfA$ a \emph{non-standard affine transformation}.
If every $A_i$ is an affine bijection, then we call $\bfA$ a \emph{non-standard affine bijection}.
Let $\Aff(\bfF)$ denote the set of all non-standard affine bijections.

Let $\bff = \lim\limits_{i\to \omega}f_i$ be a function limit and $\bfA = \lim\limits_{i \to \omega}A_i$ be a non-standard affine transformation.
Then, for any $\bfx = \lim\limits_{i \to \omega}x_i \in \bfF$,
we have $(\bff \circ \bfA)(\bfx) = \bff(\lim\limits_{i \to \omega}A_i x_i) = \st(\lim\limits_{i \to \omega}f_i(A_ix_i))$.
Hence, $\bff \circ \bfA = \st(\lim\limits_{i \to \omega}(f_i \circ A_i))$ holds.

\begin{lemma}\label{lem:stable-under-non-standard-affine-bijection}
  For any function $\bff \in \caF$, a system of linear forms $\caL$, and $\bfA \in \Aff(\bfF)$, we have
  \[
    t_\caL(\bff) = t_\caL(\bff\circ \bfA).
  \]
\end{lemma}
\begin{proof}
  Let $\bfA = \lim\limits_{i \to \omega}A_i$, where $A_i$ is an affine bijection for each $i \in \bbN$.
  Using Lemma~\ref{lem:tL-exchange} twice, we have
  \[
    t_\caL(\bff)
    = \st(\lim_{i \to \omega}t_\caL(f_i))
    = \st(\lim_{i \to \omega}t_\caL(f_i \circ A_i))
    = t_\caL(\st(\lim_{i \to \omega}(f_i \circ A_i)))
    = t_\caL(\bff \circ \bfA).
    \qedhere
  \]
\end{proof}

To identify a function $f:\bbF^n \to \bbR$ with a function limit,
we first construct a function sequence as follows:
for each $i \in \bbN$, we take an arbitrary affine transformation $A_i:\Fp^i \to \Fp^n$ with $\rank(A_i) = \min(n, i)$, and define $f_i:\bbF^i \to \bbR$ as $f_i = f \circ A_i$.
Then, we identify $f$ with $\nst{f} = \st(\lim\limits_{i \to \omega}f_i)$.
Though the choice of $\nst{f}$ is not unique,
the value $t_\caL(\nst{f})$ is uniquely determined as shown in the following lemma.
\begin{lemma}\label{lem:t-cal-equivalence}
  Let $f:\bbF^n \to \bbR$ be a bounded function and $\caL$ be a system of linear forms.
  Then, $t_\caL(\nst{f})$ is well defined and
  \[
    t_\caL(f) = t_\caL(\nst{f}).
  \]
\end{lemma}
\begin{proof}
  Let $\nst{f} = \st(\lim\limits_{i \to \omega}f \circ A_i)$ for $A_i:\Fp^i \to \Fp^n$.
  We have
  \begin{align*}
    t_\caL(\nst{f})
    = \st(\lim_{i \to \omega}t_\caL(f \circ A_i))
    = \st(\lim_{i \to \omega}t_\caL(f))
    = t_\caL(f).
  \end{align*}
  The second equality holds since $t_\caL(f \circ A_i) = t_\caL(f)$ for all $i \geq n$, and the non-principal filter $\omega$ does not contain any finite set.
\end{proof}

\section{Metric over function limits}\label{sec:metric}


Now we introduce the central notion of this paper.
The \emph{$\upsilon^d$-distance} between two function limits $\bff,\bfg \in \caF$ is defined as follows:
\[
  \upsilon^d(\bff,\bfg) = \inf_{\bfA \in \Aff(\bfF)} \|\bff - \bfg \circ \bfA\|_{U^d}.
\]
Since $\|\cdot\|_{U^d}$ is a (semi-)norm, by identifying functions with a $\upsilon^d$-distance of zero, $(\caF,\upsilon^d)$ forms a metric space.
We call this space the \emph{$\upsilon^d$-metric (space)}.
By the following lemma, we can determine the distance between two functions over different domains.
\begin{lemma}\label{lem:distance-well-defined}
  Let $f:\Fp^n \to \bit$ and $g:\Fp^m \to \bit$.
  Then, $\upsilon^d(\nst{f},\nst{g})$ is well defined.
\end{lemma}
\begin{proof}
  Suppose $\nst{f} = \st(\lim\limits_{i \to \omega}(f \circ A_i))$ and $\nst{g} = \st(\lim\limits_{i \to \omega}(g \circ B_i))$.
  Then,
  \begin{align*}
    \upsilon^d(\nst{f}, \nst{g})
    & = \inf_{\bfX \in \Aff(\bfF)}\|\nst{f} - \nst{g} \circ \bfX\|_{U^d} \nonumber \\
    & = \inf_{\bfX \in \Aff(\bfF)} \|\st(\lim_{i \to \omega}(f \circ A_i - g \circ B_i \circ X_i ))\|_{U^d}  \tag{$\bfX =: \lim\limits_{i\to\omega}X_i$} \nonumber \\
    & = \inf_{\bfX \in \Aff(\bfF)} \st(\lim_{i \to \omega}\|(f \circ A_i - g \circ B_i \circ X_i)\|_{U^d})
    \tag{by Lemma~\ref{lem:tL-exchange}}
    \label{eq:distance-well-defined}
  \end{align*}
  Let $A^*_i:\Fp^i \to \Fp^n$ and $B^*_i:\Fp^i \to \Fp^m$ be matrices that minimize $\|f \circ A^*_i - g \circ B^*_i\|_{U^d}$.
  When $i \geq \max(n, m)$, there exists an affine transformation $X^*_i:\Fp^i \to \Fp^i$ such that
  $\|f \circ A_i - g \circ B_i \circ X^*_i\|_{U^d} = \|f \circ A^*_i - g \circ B^*_i\|_{U^d}$.
  We note that for any two sequences $(a_i)$ and $(b_i)$ with $a_i \leq b_i$,
  $\st(\lim\limits_{i \to \omega}a_i) \leq \st(\lim\limits_{i \to \omega}b_i)$ holds.
  Hence,
  \begin{align*}
    \upsilon^d(\nst{f}, \nst{g}) 
    = \st(\lim_{i \to \omega}\|f \circ A_i - g \circ B_i \circ X^*_i\|_{U^d})
    = \st(\lim_{i \to \omega}\|f \circ A^*_i - g \circ B^*_i\|_{U^d}),
  \end{align*}
  which is determined regardless of the choice of $A_i$ and $B_i$.
\end{proof}

\subsection{Equivalence between $t$-convergence and $\upsilon$-convergence}

Let $(f_i:\Fp^{n_i} \to \bbR)$ be a sequence of bounded functions.
We say that the sequence is \emph{$t$-convergent} if for every finite affine system $\caL$ of liner forms,
the sequence $(t_\caL(f_i))$ converges (in the sense of Cauchy).
If there exists a function limit $\bff \in \caF$ such that $\lim\limits_{i \to \infty}t_\caL(f_i) = t_\caL(\bff)$ for every finite affine system $\caL$ of linear forms,
then we say that the sequence $(f_i)$ $t$-converges to $\bff$.
Similarly, a sequence $(\bff_i \in \caF)$ of function limits is said to be $\emph{$t$-convergent}$ if,
for every finite affine system $\caL$ of linear forms,
the sequence $(t_\caL(\bff_i))$ converges.

For a sequence $(\bff:\bfF \to \bbR)$ of function limits,
we say that it is \emph{$\upsilon$-convergent} if it is Cauchy in the $\upsilon^d$-metric for any $d \in \bbN$.
The main objective of this section is to show that $t$-convergence and $\upsilon$-convergence coincide in the following sense:
\begin{theorem}\label{the:convergent=cauchy}
  A sequence of functions $(f_i:\Fp^{n_i} \to \bit)$ is $t$-convergent to $\bff:\bfF \to \bit$ if and only if the sequence $(\nst{f_i})$ is $\upsilon$-convergent to $\bff$.
\end{theorem}
In the following two sections, we show the sufficiency (Corollary~\ref{cor:u-convergent->t-convergent}) and necessity of $\upsilon$-convergence (Corollary~\ref{cor:t-convergent->u-convergent}), respectively.


\subsubsection{$\upsilon$-convergence implies $t$-convergence}

We first look at the easier direction, that is, $\upsilon$-convergence of $(\nst{f_i})$ implies $t$-convergence of $(f_i)$.
We need the following simple proposition.
\begin{proposition}\label{pro:one-to-one}
  Let $(a_i)$ be a sequence of real numbers and $f:\bbR \to \bbR$ be a one-to-one function.
  Then, we have
  \[
    \st(\lim_{i \to \omega}f(a_i)) = f(\st(\lim_{i \to \omega}a_i)).
  \]
\end{proposition}
\begin{proof}
  Let $s = \st(\lim_{i \to \omega}f(a_i))$.
  Then, $\{i \in \bbN : f(a_i) = s \} \in \omega$ holds.
  Since $f$ is one-to-one,
  we have $\{i \in \bbN : a_i = f^{-1}(s) \}  \in \omega$.
  It follows that $f(\st(\lim\limits_{i \to \omega}a_i)) = f(f^{-1}(s)) = s$.
\end{proof}

\begin{lemma}\label{lem:tl-wrt-gowers-distance}
  Let $\bff, \bfg \in \caF_{[0,1]}$ be function limits.
  For any system of linear forms $\caL$ of true complexity of at most $d$,
  we have
  \[
    |t_\caL(\bff) - t_\caL(\bfg) | \leq \eta(\upsilon^{d+1}(\bff,\bfg)),
  \]
  where $\eta:\bbR_+ \to \bbR_+$ is a function with $\lim\limits_{\epsilon \to 0}\eta(\epsilon) = 0$.
\end{lemma}
\begin{proof}
  By Lemma~\ref{lem:stable-under-non-standard-affine-bijection},
  it suffices to show that
  $|t_\bfL(\bff) - t_\bfL(\bfg) | \leq \eta(\|\bff - \bfg\|_{U^{d+1}})$.
  Let $(f_i:\Fp^i \to [0,1])$ and $(g_i:\Fp^i \to [0,1])$ be liftings of $\bff$ and $\bfg$, respectively.
  Then, we have
  \begin{align}
  |t_\caL(\bff) - t_\caL(\bfg)|
  & =  |\st(\lim_{i \to \omega} t_\caL(f_i)) - \st(\lim_{i \to \omega} t_\caL(g_i))| \tag{by Lemma~\ref{lem:tL-exchange}} \nonumber \\
  & = \st(\lim_{i \to \omega} |t_\caL(f_i) - t_\caL(g_i)|). \label{eq:tl-wrt-gowers-distance}
  \end{align}

  Let $\delta = \|\bff - \bfg\|_{U^{d+1}}$ and $\delta_i  =\|f_i - g_i\|_{U^{d+1}}$ for each $i \in \bbN$.
  By Lemma~\ref{lem:tL-exchange}, $\delta = \st(\lim\limits_{i \to \omega}\delta_i)$.
  Since the true complexity of $\caL$ is at most $d$,
  by Lemma~\ref{lem:tl-difference-by-gowers-norm-difference},
  there exists a function $\eta:\bbR_+ \to \bbR_+$ with $\lim\limits_{\epsilon \to 0}\eta(\epsilon) = 0$ such that $|t_\caL(f_i) - t_\caL(g_i)| \leq \eta(\delta_i)$ holds for every $i \in \bbN$.
  Furthermore, we can choose $\eta$ as a strictly increasing function so that $\eta$ is one-to-one.
  From $\lim\limits_{i\to \omega}|t_\caL(f_i) - t_\caL(g_i) | \leq \lim\limits_{i\to \omega}\eta(\delta_i)$ and Proposition~\ref{pro:one-to-one},
  we have
  \begin{align*}
    \eqref{eq:tl-wrt-gowers-distance}
    & \leq \st(\lim_{i\to \omega} \eta(\delta_i) )  = \eta( \st(\lim_{i\to \omega} \delta_i) )
    = \eta(\delta).
    \qedhere
  \end{align*}
\end{proof}

\begin{corollary}\label{cor:u-convergent->t-convergent}
  Let $(f_i:\Fp^{n_i} \to \bit)$ be a sequence of function.
  If the sequence $(\nst{f_i})$ is $\upsilon$-convergent to $\bff:\bfF \to \bit$, then the sequence $(f_i)$ is $t$-convergent to $\bff$.
\end{corollary}
\begin{proof}
  If $(\nst{f_i})$ is $\upsilon$-convergent to $\bff$,
  then $t_\caL(\nst{f_i})$ converges to $t_\caL(\bff)$ for all finite affine systems $\caL$ of linear forms, by Lemma~\ref{lem:tl-wrt-gowers-distance}.
  Since $t_\caL(\nst{f_i}) = t_\caL(f_i )$ by Lemma~\ref{lem:t-cal-equivalence}, we have the desired result.
\end{proof}


\subsubsection{$t$-convergence implies $\upsilon$-convergence}

Now we turn to the other direction, that is, $t$-convergence of $(f_i)$ implies $\upsilon$-convergence of $(\nst{f_i})$.

We first show that, for any function $f:\Fp^n \to \bit$ and a random affine embedding $A:\Fp^k \to \Fp^n$ for sufficiently large $k$, two function limits $\nst{f}$ and $\nst{(f \circ A)}$ are close in the $\upsilon^d$-metric.
To this end, we need the following two lemmas.
The first says that if two sequences of polynomials $(P_1,\ldots,P_C)$ and $(Q_1,\ldots,Q_C)$ are of high rank,
then $\Gamma(P_1,\ldots,P_C)$ and $\Gamma(Q_1,\ldots,Q_C)$ cannot be distinguished in terms of the Gowers norm for any $\Gamma:\bbT^C \to \bbR$.
\begin{lemma}\label{lem:same-structure->close}
  For any $\epsilon > 0$, $C \in \bbN$, and $d \in \bbN$, there exists $r = r_{\ref{lem:same-structure->close}}(\epsilon, C, d)$ with the following property.
  Let $\Gamma:\bbT^C \to \bbR$ and let $(P_1,\ldots,P_C)$ and $(Q_1,\ldots,Q_C)$ be sequences of polynomials of degrees of at most $d$ and of ranks of at least $r$.
  Then, $\|\Gamma(P_1,\ldots,P_C) - \Gamma(Q_1,\ldots,Q_C) \|_{U^d} \leq \epsilon$ holds.
\end{lemma}
\begin{proof}
  We choose $r_{\ref{lem:same-structure->close}}(\epsilon,C,d) \geq r_{\ref{the:small-gowers-norm->high-rank}}(\epsilon/p^{dC},d)$.
  For $\gamma \in \Fp^C$, define $P_\gamma = \sum_{i \in [C]}\gamma_i P_i$.
  Note that we can write $\Gamma(P_1(x),\ldots,P_C(x)) = \sum_{\gamma \in \Fp^C}\widehat{\Gamma}(\gamma)\sfe(P_\gamma(x))$, where $\widehat{\Gamma}(\gamma)$ is the Fourier coefficient of $\Gamma$ at $\gamma$.
  Then we have
  \begin{align*}
    & \|\Gamma(P_1,\ldots,P_C) - \Gamma(Q_1,\ldots,Q_C\|_{U^d}
    =
    \Bigl\|\sum_{\gamma \in \Fp^C}\widehat{\Gamma}(\gamma) (\sfe(P_\gamma) - \sfe(Q_\gamma))\Bigr\|_{U^d} \\
    & \leq
    \| \widehat{\Gamma}(\emptyset) (\sfe(P_\emptyset) - \sfe(Q_\emptyset))\|_{U^d}
    +
    \sum_{\gamma \neq \emptyset}(
    \|\widehat{\Gamma}(\gamma)\sfe(P_\gamma) \|_{U^d}
    +
    \|\widehat{\Gamma}(\gamma)\sfe(Q_\gamma) \|_{U^d}) \\
    & \leq
    0 + \frac{\epsilon}{p^{dC}} p^{dC} = \epsilon. \tag{By Lemma~\ref{the:small-gowers-norm->high-rank}}
  \end{align*}
\end{proof}

The second lemma says that the $L_2$ and Gowers norms are preserved by extending the domain of a function through an affine transformation.
\begin{lemma}\label{lem:affine-expansion}
  Let $f:\bbF^k \to \bbR$ and $A:\bbF^n \to \bbF^k$ be an affine transformation with $n \geq k$ and $\rank(A) = k$.
  Then we have
  \begin{itemize}
  \itemsep=0pt
  \item $\|f \circ A\|_2 = \|f\|_2$
  \item $\|f \circ A\|_{U^d} = \|f\|_{U^d}$ for any $d \in \bbN$.
  \end{itemize}
\end{lemma}
\begin{proof}
  Since $A$ has rank $k$,
  the distribution of $Ax \in \bbF^k$ is uniform when $x \in \bbF^n$ is chosen uniformly at random.
  Hence $\|f \circ A\|_2 = \|f\|_2$ holds.
  Similarly,
  the distribution of $(Ax, Ay_1,\ldots,Ay_d) \in (\Fp^k)^{d+1}$ is uniform when $(x,y_1,\ldots,y_d) \in (\Fp^n)^{d+1}$ is chosen uniformly at random.
  Hence $\|f \circ A\|_{U^d} = \|f\|_{U^d}$ holds.
\end{proof}

\begin{lemma}\label{lem:distance-to-random-sample}
  Let $\epsilon > 0 $,
  $d \in \bbN$,
  and $f:\Fp^n \to \bit$ be a function.
  If $n \geq k \geq k_{\ref{lem:distance-to-random-sample}}(\epsilon,d) \in \bbN$,
  then for a random affine embedding $A:\Fp^k \to \Fp^n$,
  \[
    \upsilon^d(\nst{f}, \nst{(f \circ A)})  \leq \epsilon
  \]
  holds with a probability of at least $1-\epsilon$.
\end{lemma}
\begin{proof}
  Let $f' = f \circ A$ and let $A^+: \Fp^n \to \Fp^k$ be an affine transformation such that $A^+A = I_k$.
  Note that $\rank(A^+) = k$.
  Showing that $\|f - f'\circ A^+ \|_{U^d} \leq \epsilon$ is sufficient.
  To see this,
  for each $i \in \bbN$, let $A_i: \Fp^i \to \Fp^n$ be an arbitrary affine transformation of rank $\min(i,n)$.
  Then,
  $\nst{f}$ and $\nst{f'}$ can be chosen as $\nst{f} = \st(\lim\limits_{i \to \omega}f \circ A_i )$ and $\nst{f'} = \st(\lim\limits_{i \to \omega}f' \circ A^+ \circ A_i)$.
  (Recall that $\upsilon^d(\nst{f},\nst{f'})$ is well defined by Lemma~\ref{lem:distance-well-defined} regardless of the choice of $A_i$.)
  Now we have
  \begin{align*}
    & \upsilon^d(\nst{f},\nst{f'})
    =
    \inf_{\bfX \in \Aff(\bfF)}\|\nst{f} - \nst{f'}\circ \bfX \|_{U^d} \\
    & =
    \inf_{\bfX \in \Aff(\bfF)}\|\st(\lim_{i \to \omega}(f \circ A_i - f' \circ A^+ \circ A_i \circ  X_i))  \|_{U^d} \tag{$\bfX =: \lim\limits_{i \to \omega}X_i $}\\
    & \leq
    \|\st(\lim_{i \to \omega}(f \circ A_i - f' \circ A^+ \circ A_i))  \|_{U^d} \\
    & =
    \st(\lim_{i \to \omega}\|f \circ A_i- f' \circ A^+ \circ A_i\|_{U^d}) \tag{by Lemma~\ref{lem:tL-exchange}} \\
    & =
    \st(\lim_{i \to \omega}\|f- f' \circ A^+\|_{U^d}) \tag{by Lemma~\ref{lem:affine-expansion} and the fact that $\omega$ has no finite set}.
  \end{align*}
  Hence, if $\|f - f' \circ A^+\|_{U^d} \leq \epsilon$, then we have $\upsilon^d(\nst{f},\nst{f'}) \leq \st(\lim\limits_{i\to\omega}\epsilon) = \epsilon$.

  Let $\gamma = (\epsilon/9)^{2^d}$ and define $\eta:\bbN \to \bbR_+$ and $r:\bbN \to \bbN$ as $\eta(D) \leq \epsilon/9$ and $r(D) = r_{\ref{lem:same-structure->close}}(\epsilon/3,D, d)$, respectively.
  By applying Theorem~\ref{the:decomposition} to $f$ with these parameters, we obtain a decomposition $f = f_1 + f_2 + f_3$.
  Here, we have $f_1 = \Gamma(P)$ for some polynomial sequence $(P_1,\ldots,P_C)$, where $C \leq C_{\ref{the:decomposition}}(\gamma,\eta,d,r)$.
  Let $\caB$ be the factor defined by the polynomial sequence $(P_1,\ldots,P_C)$.

  We consider the function $f' = f'_1 + f'_2 + f'_3$, where $f'_i = f_i \circ A$ for each $i \in [3]$.
  Let $P'_i = P_i \circ A$ for each $i \in [C]$ and let $\caB'$ be the factor defined by the polynomial sequence $(P'_1,\ldots,P'_C)$.
  Note that $f_1 \circ A = \Gamma(P')$.
  Using the same argument as the proof for Claim 4.1 of~\cite{Hatami:2013ux},
  by choosing $k$ large enough as a function of $\epsilon$ and $d$,
  we have the following properties with a probability of at least $1-\epsilon$ over the choice of $A$.
  \begin{itemize}
  \itemsep=0pt
  \item $P'_i$ and $P_i$ have the same degree and depth for every $i \in [C]$.
  Moreover, $\caB'$ is $r$-regular.
  \item $\|f'_2\|_2 \leq 2\gamma$ and $\|f'_3\|_{U^d} \leq 2\eta(|\caB|)$.
  \end{itemize}

  Let $\widetilde{f} = f' \circ A^+$.
  Note that $\widetilde{f} $ can be expressed as $\widetilde{f}_1 + \widetilde{f}_2 + \widetilde{f}_3$, where $\widetilde{f}_i = f'_i \circ A^+$ for each $i\in [3]$.
  Also let $\widetilde{P}_i = P'_i \circ A^+$ for each $i \in [C]$.
  Note that $P_i$ and $\widetilde{P}_i$ have the same degree (at most $d$) and the same depth for each $i \in [C]$ since $\widetilde{P}_i = P'_i \circ A^+$ and $P'_i = \widetilde{P}_i \circ A$, and affine transformation only decreases or preserves degree and depth.
  Applying Lemma~\ref{lem:same-structure->close} to $f_1 = \Gamma(P)$ and $\widetilde{f}_1 = \Gamma(\widetilde{P})$,
  we have
  \[
    \| f_1 - \widetilde{f}_1\|_{U^d} \leq \frac{\epsilon}{3}.
  \]

  Thus,
  \begin{align*}
    \|f - \widetilde{f}\|_{U^d}
    & \leq
    \|f_1 - \widetilde{f}_1\|_{U^d} + \|f_2 - \widetilde{f}_2 \|_{U^d} + \|f_3 - \widetilde{f}_3\|_{U^d} \\
    & \leq
    \|f_1 - \widetilde{f}_1\|_{U^d} + \|f_2\|_2^{1/2^d} + \|f'_2 \circ A^+\|_{2}^{1/2^d} + \|f_3\|_{U^d} + \|f'_3 \circ A^+\|_{U^d}  \\
    & \leq
    \frac{\epsilon}{3} +  3\gamma^{1/2^d} + 3\eta(|\caB|) \tag{By Lemma~\ref{lem:affine-expansion}} \\
    & \leq \epsilon.
    \qedhere
  \end{align*}
\end{proof}

Let $f:\Fp^n \to \bbR$ be a function and $k \leq n$ be an integer.
Then, $\proj{f}{k}$ denotes a random function $f \circ A$,
where $A$ is chosen uniformly at random from an affine embedding $A:\Fp^k \to \Fp^n$.
The distribution of $\proj{f}{k}$ is determined by $\{t_\caL(f)\}$,
where $\caL$ is over all affine systems of $k$ linear forms, as shown in the following lemma.
\begin{lemma}[In the proof of Lemma~6.1 of~\cite{Hatami:2011ig}]\label{lem:approximate-by-tl}
  Let $f:\Fp^n \to [0,1]$, $\Gamma:[0,1]^k \to \bit$, and $\epsilon > 0$.
  Let $\mu$ be an arbitrary distribution over $(\Fp^n)^k$.
  If $n \geq n_{\ref{lem:approximate-by-tl}}(\epsilon,k)$,
  then the probability
  \[
    \Pr_{(x_1,\ldots,x_k) \sim \mu}[\Gamma(f(x_1),...,f(x_k)) = 1]
  \]
  can be approximated within an additive error of $\epsilon$ by a linear combination of $t_{\caL_1}(f),\ldots,t_{\caL_m}(f)$,
  where $\caL_1,\ldots,\caL_m$ are all possible affine systems of at most $k$ linear forms.
\end{lemma}

\begin{corollary}\label{cor:t-convergent->statistically-close}
  Let $\epsilon > 0$, $d \in \bbN$, and $k \in \bbN$.
  There exist $n_{\ref{cor:t-convergent->statistically-close}}(\epsilon,d,k)$, $k' = k'_{\ref{cor:t-convergent->statistically-close}}(k)$, and $\delta = \delta_{\ref{cor:t-convergent->statistically-close}}(\epsilon,d,k)$ such that the following holds.
  Let $f:\Fp^n \to \bit$ and $g:\Fp^m \to \bit$ be functions with $\min(n,m) \geq n_{\ref{cor:t-convergent->statistically-close}}(\epsilon,d,k)$.
  If $|t_\caL(f) - t_\caL(g)| \leq \delta$ for any affine system $\caL$ of $k'$ linear forms, then the distributions $\proj{f}{k}$ and $\proj{g}{k}$ have a statistical distance of at most $\epsilon$.
\end{corollary}
\begin{proof}
  For a function $h:\Fp^{k} \to \bit$,
  define the characteristic function $\Gamma_h:\bit^{p^{k}} \to \bit$ of $h$ as
  \[
    \Gamma_h(\{a_x\}_{x \in \Fp^{k}}) = \begin{cases}
      1 & \text{if } a_x = h(x), \\
      0 & \text{otherwise}.
    \end{cases}
  \]
  We choose $n_{\ref{cor:t-convergent->statistically-close}}(\epsilon,d,k) \geq n_{\ref{lem:approximate-by-tl}}(\epsilon/(4 \cdot 2^{p^{k}}), k)$.
  Then by Lemma~\ref{lem:approximate-by-tl},
  the probability that $\proj{f}{k}$ coincides with $h$ can be approximated as follows.
  \[
    \Pr[\proj{f}{k} = h] = \Pr_{x_0,x_1,\ldots,x_k \in \Fp^n}[\Gamma_h(\{f(x_0 + \sum_{i \in [k]}b_i x_i) : b_1,\ldots,b_k \in \Fp\}) = 1] = \sum_\caL \beta_\caL  t_\caL(f) \pm \frac{\epsilon}{4 \cdot 2^{p^{k}}},
  \]
  where $\caL$ is over all possible affine systems of $p^{k}$ linear forms.
  Then,
  \[
    |\Pr[\proj{f}{k} = h]  - \Pr[\proj{g}{k} = h]| \leq \sum_\caL \beta_\caL |t_\caL(f)- t_\caL(g)| \pm \frac{\epsilon}{2 \cdot 2^{p^{k}}}.
  \]
  Let $N = N(k)$ be the number of all possible affine systems of $p^{k}$ linear forms.
  By choosing $\delta_{\ref{cor:t-convergent->statistically-close}}(\epsilon,d,k) = \epsilon / N$ and $k'_{\ref{cor:t-convergent->statistically-close}} = p^{k}$,
  the statistical distance between $\proj{f}{k}$ and $\proj{g}{k}$ becomes at most $(2^{p^{k}} \cdot \epsilon/2^{p^{k}} + \delta N)/2  = \epsilon$.
\end{proof}

We can finally show that $t$-convergence implies $\upsilon$-convergence.
\begin{lemma}\label{lem:gowers-distance-wrt-tl}
  Let $\epsilon  > 0$ and $d \in \bbN$.
  There exists $n_{\ref{lem:gowers-distance-wrt-tl}}(\epsilon,d)$,
  $k = k_{\ref{lem:gowers-distance-wrt-tl}}(\epsilon,d)$,
  and $\delta = \delta_{\ref{lem:gowers-distance-wrt-tl}}(\epsilon,d)$ such that the following holds.
  Let $f: \Fp^n \to \bit$ and $g:\Fp^m \to \bit$ be functions with $n \geq n_{\ref{lem:gowers-distance-wrt-tl}}(\epsilon,d)$.
  If $|t_\caL(f) - t_\caL(g)| \leq \delta$ for any affine system $\caL$ of $k$ linear forms,
  then we have
  \[
    \upsilon^{d}(\nst{f},\nst{g}) \leq \epsilon.
  \]
\end{lemma}
\begin{proof}
  Let $k' = k_{\ref{lem:distance-to-random-sample}}(\epsilon/3,d)$, and set $n_{\ref{lem:gowers-distance-wrt-tl}}(\epsilon,d ) \geq k'$.
  By Lemma~\ref{lem:distance-to-random-sample}, we have $\upsilon^{d}(\nst{f},\nst{(\proj{f}{k'})}) \leq \epsilon/3$ and $\upsilon^{d}(\nst{g},\nst{(\proj{g}{k'})}) \leq \epsilon/3$ with a probability of at least $1-\epsilon/3$.

  Now we consider the distance $\upsilon^{d}(\nst{(\proj{f}{k'})}, \nst{(\proj{g}{k'})})$.
  We set $\delta_{\ref{lem:gowers-distance-wrt-tl}}(\epsilon,d) = \delta_{\ref{cor:t-convergent->statistically-close}}(\epsilon/3,d,k')$ and $k_{\ref{lem:gowers-distance-wrt-tl}}(\epsilon,d)  = k_{\ref{cor:t-convergent->statistically-close}}(k')$ and $n_{\ref{lem:gowers-distance-wrt-tl}}(\epsilon,d ) \geq n_{\ref{cor:t-convergent->statistically-close}}(\epsilon/3,d,k')$.
  By Corollary~\ref{cor:t-convergent->statistically-close},
  the statistical distance between $\proj{f}{k'}$ and $\proj{g}{k'}$ is at most $\epsilon/3$.
  Hence, we can couple $\proj{f}{k'}$ and $\proj{g}{k'}$ so that $\proj{f}{k'} = \proj{g}{k'}$ holds with a probability of at least $1 - \epsilon/3$.

  By the union bound, these events happen simultaneously with a probability of at least $1-\epsilon$.
  Hence, there exist affine embedding $A:\Fp^k \to \Fp^n$ and $A': \Fp^k \to \Fp^m$ such that
  \begin{align*}
    \upsilon^d(f,g)
    & \leq
    \upsilon^d(\nst{f},\nst{(f\circ A)}) + \upsilon^d(\nst{(f\circ A)}, \nst{(g\circ A')}) + \upsilon^d(\nst{g},\nst{(g\circ A')}) \tag{by the triangle inequality} \\
    & \leq \frac{\epsilon}{2} + 0 + \frac{\epsilon}{2} = \epsilon,
  \end{align*}
  which implies the lemma.
\end{proof}

\begin{corollary}\label{cor:t-convergent->u-convergent}
  Let $(f_i:\Fp^{n_i} \to \bit)$ be a sequence of functions.
  If the sequence $(f_i)$ is $t$-convergent to $\bff:\bfF \to \bit$, then the sequence $(\nst{f_i})$ is $\upsilon$-convergent to $\bff$.
\end{corollary}
\begin{proof}
  If the sequence $(f_i)$ is $t$-convergent to $\bff$, then for any finite affine system $\caL$ of linear forms, $(t_\caL(f_i))$ is convergent to $t_\caL(\bff)$.
  Hence, by Lemma~\ref{lem:gowers-distance-wrt-tl}, $(\nst{f_i})$ converges to $\bff$ in the $\upsilon^d$-metric for every $d \in \bbN$, which means that $(\nst{f_i})$ is $\upsilon$-convergent to $\bff$.
\end{proof}

\subsection{Other properties}
This section discusses other properties of the $\upsilon^d$-metrics.
First, we show that any function limit can be realized as a limit of functions in terms of $t$-convergence.
\begin{lemma}\label{lem:realize-as-limit}
  For any function limit $\bff: \bfF \to \bit$, there exists a sequence of functions $(f_i:\Fp^{n_i} \to \bit)$ that $t$-converges to $\bff$.
\end{lemma}
\begin{proof}
  Let $(f_i:\Fp^i \to \bit)$ be a lifting of $\bff$.
  Then for any system of linear forms $\caL$,
  we have $t_\caL(\bff) = \st(\lim\limits_{i \to \omega}t_\caL(f_i))$ by Lemma~\ref{lem:tL-exchange}.
  This means that the set $I_\caL := \{i \in \bbN : t_\caL(f_i) = t_\caL(\bff)\}$ is contained in $\omega$.
  Consider an arbitrary order $\caL_1, \caL_2,\ldots$ of all possible finite affine systems of linear forms.
  We inductively construct a sequence of integers $I^k = (i^k_1,i^k_2,\ldots)$ for each integer $k \geq 0$  as follows.
  First, we set $I^0 = \bbN$.
  Then, for each $k \in \bbN$, we define $I^k = I^{k-1} \cap I_{\caL_k}$.
  Note that $I^k \in \omega$ since a filter is closed under taking intersections.
  Furthermore, since $\omega$ is a non-principal filter, $I^k$ is an infinite sequence of integers.
  Let $(f'_j)$ be the sequence of functions defined as $f'_j = f_{i^j_j}$.
  For any $k \in \bbN$ and $j \geq k$,
  we have $t_{\caL_k}(f'_j) = t_{\caL_k}(\bff)$.
  Hence, the sequence $(f'_j)$ $t$-converges to $\bff$.
\end{proof}
From Theorem~\ref{the:convergent=cauchy}, this also means that any function limit $\bff:\bfF \to \bit$ has a sequence of functions $(f_i:\Fp^{n_i} \to \bit)$ such that the sequence $(\nst{f_i})$ $\upsilon$-converges to $\bff$.

Next, to show that the $\upsilon^d$-metric is compact for any $d \in \bbN$,
we show the following, stronger, property.
\begin{lemma}\label{lem:strong-compact}
  Let $(\bff_i \in \caF_{\bit})_{i \in \bbN}$ be a sequence of function limits.
  Then, there exists a subsequence of $(\bff_i)$ that $\upsilon$-converges.
\end{lemma}
\begin{proof}
  Let $(\bff^i)$ be a sequence of function limits in $\caF_{\bit}$.
  We want to construct a subsequence that has a limit in $\caF_{\bit}$.
  First, we construct a subsequence that $t$-converges as follows.
  Consider an arbitrary order $\caL_1,\caL_2,\ldots$ of all possible finite affine systems of linear forms.
  Define a sequence $(\bfg^0_i)$ by $\bfg^0_i = \bff_i$ for each $i \in \bbN$.
  Then, for each $k \in \bbN$,
  we inductively define a sequence $(\bfg^k_i)$ as a subsequence of $(\bfg^{k-1}_i)$ so that $(t_{\caL_k}(\bfg^k_i))$ converges.
  This is possible since the metric space $([-1,1], \ell_1)$ is compact.
  Finally, we define a sequence $(\bfg_i)$ of function limits as $\bfg_i = \bfg^i_i$ for each $i \in \bbN$.
  We can observe that $(\bfg_i)$ is a subsequence of $(\bff_i)$ and $t$-converges.
  Now we replace $(\bff^i)$ with $(\bfg^i)$ and assume that $(\bff^i)$ is a sequence of function limits that $t$-converges.

  By Lemma~\ref{lem:realize-as-limit},
  for each $i \in \bbN$,
  we can take a function sequence $(f^i_j:\Fp^j \to \bit)$ that $t$-converges to $\bff^i$ and, hence, $\upsilon$-converges to $\bff^i$.
  Now, we construct a sequence $(g_i:\Fp^i \to \bit)$ by first setting $g_1 = f^1_1$, and then, for each $i \in \bbN$, inductively defining $g_i$ from $g_{i-1}$ as follows:
  first, choose an index $k_i$ so that $|t_{\caL_{i'}}(f^i_{k_i}) - t_{\caL_{i'}}(\bff^i)| \leq |t_{\caL_{i'}}(g_{i-1}) - t_{\caL_{i'}}(\bff^{i-1})|/2$ and $\upsilon^{i'}(\nst{f^i_{k_i}}, \bff^i) \leq \upsilon^{i'}(\nst{g_{i-1}}, \bff^{i-1})/2$ hold for every $i' \leq i$ (we can choose such $k_i$ since $(f^i_j)$ $t$-converges and, hence, $\upsilon$-converges to $\bff^i$),
  then we set $g_i = f^i_{k_i}$.
  This gives us (i) $\lim\limits_{i \to \infty}|t_{\caL_k}(g_i) - t_{\caL_k}(\bff^i)| = 0$ for any $k \in \bbN$, and (ii) $\lim\limits_{i \to \infty}\upsilon^d(\nst{g_i}, \bff^i) = 0$ for any $d \in \bbN$.
  Since the sequence $(\bff^i)$ $t$-converges,
  the sequence $(g_i)$ also $t$-converges by (i).
  By Theorem~\ref{the:convergent=cauchy}, there exists a function limit $\bfg:\bfF \to \bit$ to which $(\nst{g_i})$ $\upsilon$-converges.
  Hence, the sequence $\bff^i$ $\upsilon$-converges to $\bfg$, by (ii).
\end{proof}

\begin{corollary}\label{cor:weak-compact}
  The metric space $(\caF_{\bit},\upsilon^d)$ is compact for any $d \in \bbN$.
\end{corollary}
\begin{proof}
  For any function sequence $(\bff_i:\caF \to \bit)$,
  there exists a subsequence that $\upsilon$-converges.
  In particular, it converges in the $\upsilon^d$-metric.

\end{proof}

\section{Characterization of Estimable Parameters}\label{sec:characterization}

Let $\pi$ be an affine-invariant function parameter,
that is, for each function of the form $f:\Fp^n \to \bit$, $\pi$ associates a value $\pi(f) \in [0,1]$.
This section gives a characterization of obliviously constant-query estimable affine invariant properties, using the tools developed in previous sections.

The following theorem gives a number of equivalent conditions characterizing the  testability of a function parameter.
\begin{theorem}\label{the:characterization}
  Let $\pi$ be an affine-invariant parameter with $\pi \in [0,1]$ that is defined over functions of the form $f:\Fp^n \to \bit$.
  The following are equivalent:
  \begin{itemize}
  \setlength{\itemsep}{0pt}
  \item[(a)] $\pi$ is obliviously constant-query estimable.
  \item[(b)] There exists a function parameter $\widetilde{\pi}$, possibly different from $\pi$, with the following property.
    For every $\epsilon > 0$ and sufficiently large $k$, every function $f:\Fp^n \to \bit$ with $n \geq k$ satisfies $|\pi(f) - \E[\widetilde{\pi}(f \circ A)]| < \epsilon$ for a random affine embedding $A:\Fp^k \to \Fp^n$.
  \item[(c)] For every $t$-convergent sequence $(f_i:\Fp^{n_i} \to \bit)$, the sequence of numbers $(\pi(f_i))$ is convergent.
  \item[(d)] There exists a functional $\widehat{\pi}(\cdot)$ on $\caF_{\bit}$ with the following properties:
    (i) $\widehat{\pi}$ is continuous in the sense that,
      for any sequence $(f_i:\Fp^{n_i} \to \bit)$ of functions such that $(\nst{f_i})$ $\upsilon$-converges to $\bff$,
      $\lim\limits_{i \to \infty}\widehat{\pi}(\nst{f_i}) = \widehat{\pi}(\bff)$ holds.
    (ii) $\widehat{\pi}$ extends $\pi$ in the sense that $\widehat{\pi}(\nst{f}) = \pi(f)$.


  \end{itemize}
\end{theorem}
\begin{proof}
  (a) $\Rightarrow$ (b):
  The definition of oblivious constant-query estimability is very similar to condition~(b);
  it states that a random affine embedding $A:\Fp^k \to \Fp^n$, as in (b),  satisfies
  \[
    |\pi(f) - \widetilde{\pi}(f \circ A)| < \epsilon
  \]
  with large probability, where $\widetilde{\pi}(\cdot)$ is taken to be the output of the algorithm.
  This clearly implies that this difference is small on  average.

  (b) $\Rightarrow$ (c):
  Suppose that a sequence $(f_i:\Fp^{n_i} \to \bit)$ is $t$-convergent.
  By Corollary~\ref{cor:t-convergent->statistically-close},
  for sufficiently large $j,j'\in \bbN$,
  the distribution of $\proj{f_j}{k}$ is very close to the distribution of $\proj{f_{j'}}{k}$.
  Hence, $|\E[\widetilde{\pi}(\proj{f_j}{k})]] - \E[\widetilde{\pi}(\proj{f_{j'}}{k})]| \leq \epsilon/3$.
  By (b), we can choose a large enough $k$ so that $|\pi(f_j) - \E[\widetilde{\pi}(\proj{f_j}{k})]| \leq \epsilon/3$ and $|\pi(f_{j'})-\E[ \widetilde{\pi} (\proj{f_{j'}}{k})]| \leq \epsilon/3$ hold, and so $|\pi(f_j) - \pi(f_{j'})| \leq \epsilon$ holds.

  (c) $\Rightarrow$ (a):
  If condition~(a) fails to hold,
  then there exists $\epsilon > 0$ such that,
  for infinitely many $k$,
  there exists a function $f:\Fp^{n} \to \bit$ for which $|\pi(f) - \widetilde{\pi}(\proj{f}{k})| \geq \epsilon$ holds with a probability of at least $1/3$ for any function parameter $\widetilde{\pi}$.
  In particular, we can choose $\widetilde{\pi} = \pi$.
  Let $(k_i)$ and $(f_i:\Fp^{n_i} \to \bit)$ be the sequences of such $k$'s and $f$'s.
  By taking the subsequence,
  we may assume that $k_i \geq k_{\ref{lem:distance-to-random-sample}}(1/i,i)$.
  Further, by Lemma~\ref{lem:strong-compact},
  we may assume that the sequence $(\nst{f_i})$ is $\upsilon$-convergent.
  By Theorem~\ref{lem:distance-to-random-sample},
  $\upsilon^{i}(\nst{f_i},\nst{(\proj{f_i}{i})}) \leq 1/i$ with a probability of at least $1-1/i$.
  Hence, we can fix $A_i:\Fp^i \to \Fp^{n_i}$ such that both
  \begin{align}
    |\pi(f_i) - \pi(f_i \circ A_i)| \geq \epsilon \label{eq:(c)->(a)-1}
  \end{align}
  and
  \begin{align}
    \upsilon^{i}(\nst{f_i} , \nst{(f_i \circ A_i)}) \leq 1/i \label{eq:(c)->(a)-2}
  \end{align}
  hold.

  Now merging the sequences $(\nst{f_i})$ and $(\nst{f_i \circ A_i})$,
  we get a $\upsilon$-convergent sequence by~\eqref{eq:(c)->(a)-2}.
  By Theorem~\ref{the:convergent=cauchy}, this sequence is $t$-convergent.
  However, condition~(c) is violated by~\eqref{eq:(c)->(a)-1}.

  (c) $\Rightarrow$ (d):
  Consider any $\bff :\bfF \to \bit$.
  By Lemma~\ref{lem:realize-as-limit},
  there exists a sequence of functions that $t$-converges to $\bff$.
  Let $(f_i)$ be any such sequence and define $\widehat{\pi}(\bff)$ as the limit of $\pi(f_i)$.
  From condition~(c), this value does not depend on the choice of the sequence.
  From the construction, $\widehat{\pi}$ satisfies property~(i).
  To see property (ii),
  consider the sequence consisting only of the same function $f$, which $t$-converges to $\nst{f}$ by Lemma~\ref{lem:t-cal-equivalence}.
  Then, $\widehat{\pi}(\nst{f})$ is equal to the limit of the sequence consisting only of the same value $\pi(f)$, which is $\pi(f)$.

  (d) $\Rightarrow$ (c):
  Consider a $t$-convergent sequence $(f_i:\Fp^{n_i} \to \bit)$ and let $\bff \in \caF_{\bit}$ be its limit.
  Then, $(\nst{f_i})$ is $\upsilon$-convergent to $\bff$ by Theorem~\ref{the:convergent=cauchy}.
  Hence, by property~(i) of condition~(d),
  we have $\lim\limits_{i \to \infty}\widehat{\pi}(\nst{f_i}) = \widehat{\pi}(\bff)$.
  From property~(ii) of condition~(d), the sequence $(\pi(f_i))$ is also convergent to $\widehat{\pi}(\bff)$.
\end{proof}

\section{Applications}\label{sec:applications}

In this section, we apply our characterization to show Theorem~\ref{the:application}.

For a property of functions $\caP$,
let $\|f\|_{\caP}$ denote the \emph{distance to $\caP$},
that is,
\[
  \|f\|_{\caP} := \min_{g \in \caP} \|f-P\|_1.
\]


The following lemma rephrases the convergence in the $\upsilon^d$-metric without using function limits.
\begin{lemma}\label{lem:U^d-converges}
  Let $d \in \bbN$ be an integer and $(f_i:\Fp^{n_i} \to \bit)$ be a sequence of functions such that $(\nst{f_i})$ converges in the $\upsilon^d$-metric.
  Then, for any $\epsilon > 0$ and sufficiently large integers $i < j$,
  there exist some integer $n \geq \max(n_i,n_j)$ and affine transformations $A_i:\Fp^n \to \Fp^{n_i}$ and $A_j: \Fp^n \to \Fp^{n_j}$ such that
  $\|f_i \circ A_i - f_j \circ A_j\|_{U^d} \leq \epsilon$ holds.
\end{lemma}
\begin{proof}
  For any $\epsilon > 0$ and sufficiently large integers $i < j$,
  we have $\upsilon^d(\nst{f_i}, \nst{f_j}) \leq \epsilon$ by the $\upsilon$-convergence of the sequence $(f_i)$.
  Hence, there exists some non-standard affine bijection $\bfX:  \bfF \to \bfF$ such that
  $\|\nst{f_i} - \nst{f_j} \circ \bfX\|_{U^d} \leq \epsilon$ holds.
  Suppose $\bfX = \lim\limits_{k \to \omega}X_k$ for some affine bijections $(X_k:\Fp^k \to \Fp^k)$.
  Also, suppose we have $\nst{f_i} = \st(\lim\limits_{k \to \omega}f_i \circ B_k)$ and $\nst{f_j} = \st(\lim\limits_{k \to \omega}f_j \circ C_k)$ for some $B_k:\Fp^k \to \Fp^{n_i}$ and $B_k:\Fp^k \to \Fp^{n_j}$.
  Then, $\nst{f_i} - \nst{f_j} \circ \bfX = \st(\lim\limits_{k \to \omega}(f_i \circ B_k - f_j \circ C_k \circ X_k))$ holds.
  By Lemma~\ref{lem:tL-exchange}, we have
  \[
    \st(\lim\limits_{k \to \omega}\|f_i \circ B_k - f_j \circ C_k \circ X_k\|_{U^d}) = \|\nst{f_i} - \nst{f_j} \circ \bfX\|_{U^d} \leq \epsilon.
  \]
  Hence $\{ k \in \bbN \mid \|f_i \circ B_k - f_j \circ C_k \circ X_k\|_{U^d} \leq \epsilon \} \in \omega$ holds.
  Note that this set is not finite since $\omega$ is a non-principal filter.
  In particular, there exists some $n \geq \max(n_i, n_j)\in \bbN$ such that $\|f_i \circ B_n - f_j \circ C_n \circ X_n\|_{U^d} \leq \epsilon$ holds, and we have the lemma with $A_i = B_n$ and $A_j = C_n \circ X_n$.
\end{proof}

The following lemma states that, if a property $\caP$ is closed under blowing-up, then the distance to $\caP$ is preserved by blowing-up.
\begin{lemma}\label{lem:polynomial-distance}
  Let $\caP$ be a property closed under blowing-up, $f:\Fp^n \to \bit$ be a function, and $A:\Fp^m \to \Fp^n$ be an affine transformation with $m \geq n$ and $\rank(A) = n$.
  Then, we have
  \[
    \|f\|_{\caP} = \|f \circ A\|_{\caP}.
  \]
\end{lemma}
\begin{proof}
  Let $f' = f \circ A$, and let $g:\Fp^n \to \bit$ and $g':\Fp^m \to \bit$ be functions satisfying $\caP$ closest to $f$ and $f'$, respectively.
  First, we have
  \[
    \|f\|_{\caP} = \|f - g\|_1 = \|f \circ A - g \circ A\|_1
    \geq \|f' - g'\|_1 = \|f'\|_{\caP}.
  \]
  The second equality holds since $g \circ A \in \caP$ and the distribution of $Ax \in \Fp^n$ is uniform when $x \in \Fp^m$ is sampled uniformly.

  Now we show the other direction.
  In what follows, we assume that $A$ is a linear transformation of the form
  \[
    A = \left(\begin{matrix}
    I_n & O
    \end{matrix}\right).
  \]
  We can easily handle the general case by applying an appropriate affine transformation to $f$.

  Let $\caA^+: \Fp^n \to \Fp^m$ be the set of all linear transformations $A^+$ satisfying $AA^+ = I_n$.
  Note that every $A^+ \in \caA^+$ is of the form
  \[
    A^+ = \left(\begin{matrix}
      I_n \\
      B
    \end{matrix}\right),
  \]
  where $B \in \Fp^{(m - n) \times n}$ is an arbitrary matrix.

  Recall that
  \begin{align}
    \|f'\|_{\caP} = \|f \circ A - g'\|_1
    = \E_{x \in \Fp^m}[ |(f \circ A)(x) - g'(x)|]. \label{eq:polynomial-distance}
  \end{align}
  Note that $(f \circ A)(x)$ only depends on $x_1,\ldots,x_n$.
  If we fix $x_1,\ldots,x_n$ and choose $x_{n+1},\ldots,x_m$ uniformly at random,
  then the distribution of $A^+x$ is uniform over the set $\{ y \in \Fp^m : y_1 = x_1,\ldots,y_n = x_n\}$.
  Hence,
  \[
    \eqref{eq:polynomial-distance}
    =
    \E_{x \in \Fp^n} \E_{A^+ \in \caA^+}[| (f \circ A \circ A^+)(x) - (g' \circ A^+)(x)|]
  \]
  Hence, there exists some $A^+ \in \caA^+$ such that
  \[
    \|f'\|_{\caP} \geq \E_{x \in \Fp^n}[| (f \circ A \circ A^+)(x) - (g' \circ A^+)(x)|].
  \]
  However, the right hand side can be expressed as follows:
  \[
    \E_{x \in \Fp^n}[| f(x) - (g' \circ A^+)(x)|]
    \geq
    \E_{x \in \Fp^n}[| f(x) - g(x)|]
    = \|f\|_{\caP}.
    \qedhere
  \]
\end{proof}

Let us define $\sfe_p : \bbF_p \to \bbC$ as $\sfe_p(x) = e^{\frac{2\pi i x}{p}}$.
Note that a function $f:\Fp^n \to \bit$ is $\epsilon$-far from $\caP$ if and only if $\|f\|_{\caP} \geq \epsilon$ .
Hence, the following theorem implies Theorem~\ref{the:application}.
\begin{theorem}\label{the:poly-is-estimable}
  Let $\caP$ be a property as in Theorem~\ref{the:application}.
  Then, the distance $\|\cdot\|_{\caP}$ is obliviously constant-query estimable.
\end{theorem}
\begin{proof}
  Let $(f_i:\Fp^{n_i} \to \bit)$ be a $t$-convergent sequence of functions.
  Then, we show that the sequence $\|f_i\|_{\caP}$ converges.
  By~(c) of Theorem~\ref{the:characterization}, this means that $\|\cdot\|_{\caP}$ is obliviously constant-query testable.

  Let $\overline{d} = \max\{d,\log\frac{1}{\epsilon}/\log \frac{p}{p-1}\}$.
  By Lemma~\ref{lem:U^d-converges},
  for any $\epsilon > 0$, for sufficiently large $i < j$,
  $\|f_i \circ A_i - f_j \circ A_j\|_{U^{\overline{d}+1}} \leq \epsilon$ holds for some affine transformations $A_i:\Fp^n \to \Fp^{n_i}$ and $A_j: \Fp^n \to \Fp^{n_j}$.
  Let $f'_i = f_i \circ A_i$ and $f'_j = f_j \circ A_j$.
  Then we have $|\E\limits_x[(f'_i(x) - f'_j(x))\sfe_p(P(x))]| \leq \epsilon$ for any polynomial $P:\Fp^n \to \Fp$ of degree at most $\overline{d}$.
  Hence, for any polynomial sequence $(P_1,\ldots,P_c)$ of degree at most $\overline{d}$ and a function $\Gamma:\Fp^c \to \bit$, we have
  \[
    \Bigl|\E_x[(f'_i(x) - f'_j(x)) \sfe(\Gamma(P_1(x),\ldots,P_c(x)))]\Bigr|
    =
    \Bigl|\sum_{\gamma \in \Fp^c} \widehat{\sfe \circ \Gamma}(\gamma) \E_x [(f'_i(x) - f'_j(x))\sfe_p(P_\gamma(x))]\Bigr|
    \leq
    p^c \epsilon,
  \]
  where $P_\gamma = \sum\limits_{k \in [c]}\gamma_k P_k$.
  Then we have $\E\limits_x[f'_i(x) \sfe(\Gamma(P_1(x),\ldots,P_c(x)))] = \E\limits_x[f'_j(x) \sfe(\Gamma(P_1(x),\ldots,P_c(x)))] \pm p^c \epsilon$.
  By some calculation, we can confirm that this implies $\|f'_i - \Gamma(P_1,\ldots,P_c)\|_1 = \|f'_j - \Gamma(P_1,\ldots,P_c)\|_1 \pm p^{c}\epsilon$.

  Let $g_j:\Fp^{n_j} \to \bit$ be the function satisfying $\caP$ that is closest to $f_j$.
  Note that $g_j$ is of the form $g_j = \Gamma(P_1,\ldots,P_c,Q_1,\ldots,Q_{c'})$ for some polynomials $P_k:\Fp^{n_j} \to \Fp$ of degree at most $d$, factored polynomials $Q_k:\Fp^{n_j} \to \Fp$, and a function $\Gamma:\Fp^{c+c'} \to \bit$.
  Let $h_j:\Fp^{n_j} \to \Fp$ be a function obtained from $g_j$ by replacing each factored polynomial $Q_k$ of degree more than $\overline{d}$ with $Q_k(0)$.
  Since we have $\Pr[Q_k(x) \neq Q_k(0)] \leq \epsilon$ if $Q_k$ has degree more than $\overline{d}$, we have $\|g_j - h_j\|_1 \leq c' \epsilon$.
  Then,
  \begin{align*}
    \|f_j\|_{\caP}
    & = \|f_j - g_j\|_1
    \geq \|f_j - h_j\|_1 - C\epsilon
    = \|f'_j - h_j \circ A_j\|_1 - C\epsilon \\
    & \geq \|f'_i - h_j \circ A_j\|_1 - (C+p^{2C})\epsilon
    \geq \|f'_i\|_{\caP} - (C + p^{2C})\epsilon
    = \|f_i\|_{\caP} - (C+p^{2C})\epsilon.
  \end{align*}
  The last equality follows by the fact that $\caP$ is closed under blowing-up and Lemma~\ref{lem:polynomial-distance}.

  Let $g_i:\Fp^{n_i} \to \bit$ be the function satisfying $\caP$ that is closest to $f_i$, and let $h_i:\Fp^{n_i} \to \Fp$ be a function obtained from $g_i$ by replacing each factored polynomial $Q$ of degree more than $\overline{d}$ with $Q(0)$.
  Similarly,
  \begin{align*}
    \|f_i\|_{\caP}
    & = \|f_i - g_i\|_1
    = \|f_i - h_i\|_1 - C\epsilon
    = \|f'_i - h_i \circ A_i\|_1 - C\epsilon \\
    & \geq \|f'_j - h_i \circ A_i\|_1 - (C+p^{2C})\epsilon
    \geq \|f'_j\|_{\caP} - (C+p^{2C})\epsilon
    = \|f_j\|_{\caP} - (C+p^{2C})\epsilon.
  \end{align*}
  Hence, we have $\|f_i\|_{\caP} = \|f_j\|_{\caP} \pm (C+p^{2C})\epsilon$,
  and the sequence $(\|f_i\|_{\caP})_{i \in \bbN}$ converges.
\end{proof}

\section{Conclusions}\label{sec:conclusion}

This work defines a metric over function limits that is based on the Gowers norm.
Properties of the metric are analyzed, and a characterization is given (Theorem~\ref{the:characterization}) of obliviously constant-query estimable parameters in terms of that metric.
This characterization is satisfactory in the sense that it is easier to understand than the one recently given by the author~\cite{Yoshida:2014tq}.
Having said that, there are several problems worth studying:
\begin{itemize}
\itemsep=0pt
\item Can we use our characterization of constant-query estimability to show that other specific parameters are constant-query estimable?
\item Can we give a characterization of properties that is constant-query testable with one sided error in terms of the $\upsilon^d$-metric?
In particular, can we prove or disprove the conjecture by~\cite{Bhattacharyya:2010gb}, which says that every affine subspace hereditary property is constant-query testable?
\item Graph limits have been used to study extremal graph theory (see~\cite{Lovasz:2012wn} for a survey). Can we use the notion of function limits to study ``extremal function theory''? A typical problem would ask how many ones a function $f:\Fp^n \to \bit$ can have when it avoids a certain pattern in its affine restriction.
\end{itemize}


\bibliographystyle{abbrv}
\bibliography{limit-object}

\end{document}